\documentclass[11pt,letterpaper]{article}
\usepackage{epsfig,rotating,setspace,latexsym,amsmath,epsf,amssymb,bm}
\usepackage{cite,graphicx,authblk,color}
\usepackage[title]{appendix}
\usepackage{amsthm}
\usepackage{comment}
\usepackage{indentfirst}
\usepackage{comment}
\usepackage{graphics}
\usepackage{float}
\usepackage[caption=false]{subfig}
\usepackage{graphicx}
\usepackage{caption}
\usepackage{textcomp}
\usepackage{xparse}
\usepackage{bbm}

\usepackage{booktabs}
\usepackage{pgfgantt}
 \usepackage{mathtools}
\usepackage{float}
\usepackage{wrapfig}
\usepackage{bbm}
\usepackage{mathalfa}
\newtheorem{definition}{Definition}
\newtheorem{theorem}{Theorem}
\newtheorem{lemma}{Lemma}

\newtheorem{remark}{Remark}

\newtheorem{proposition}{Proposition}

\usepackage{stackengine}

\title{Generalization Bounds for Neural Belief Propagation Decoders\vspace{10pt}
\author{Sudarshan Adiga, Xin Xiao, Ravi Tandon, Bane Vasi\'c, Tamal Bose}
\affil{Department of Electrical and Computer Engineering\\
University of Arizona, Tucson, AZ, USA.\\
E-mail: {\{adiga, 7xinxiao7, tandonr, vasic, tbose\}}@arizona.edu}}

\begin{document}
\maketitle
\newcommand\blfootnote[1]{%
  \begingroup
  \renewcommand\thefootnote{}\footnote{#1}%
  \addtocounter{footnote}{-1}%
  \endgroup
}

\begin{abstract}
\noindent  Machine learning based approaches are being increasingly used for designing decoders for next generation communication systems. One widely used framework is \text{neural belief propagation} (NBP), which unfolds the belief propagation (BP) iterations into a deep neural network and the parameters are trained in a data-driven manner. NBP decoders have been shown to improve upon classical decoding algorithms. In this paper, we investigate the generalization capabilities of NBP decoders. Specifically, the generalization gap of a decoder is the difference between empirical and expected bit-error-rate(s). We present new theoretical results which bound this gap and show the dependence on the \textit{decoder complexity}, in terms of code parameters (blocklength, message length, variable/check node degrees), decoding iterations, and the training dataset size. Results are presented for both regular and irregular parity-check matrices. To the best of our knowledge, this is the first set of theoretical results on generalization performance of neural network based decoders. 
We present experimental results to show the dependence of generalization gap on the training dataset size, and  decoding iterations for different codes. 
\end{abstract}

\section{Introduction}
\footnote{This work was supported by NSF grants CAREER 1651492, CCF-2100013, CNS-2209951, CNS-1822071, CIF- 1855879, CCSS-2027844, CCSS-2052751, and NSF-ERC 1941583. Bane Vasi\'{c} was also supported by the Jet Propulsion Laboratory, California Institute of Technology, under a contract with the National Aeronautics and Space Administration and funded through JPL’s Strategic University Research Partnerships (SURP) Program. Bane Vasi\'{c} has disclosed an outside interest in Codelucida to the University of Arizona. Conflicts of interest resulting from this interest are being managed by The University of Arizona in accordance with its policies. This work was presented in part at the 2023 IEEE International Symposium on Information Theory \cite{adiga2023generalization}. 
}
\noindent Deep neural networks have emerged as an important tool in 5G and beyond for hybrid beamforming \cite{li2019deep, huang2018unsupervised, peken2020deep}, channel encoding, decoding, and estimation \cite{nachmani2016learning, nachmani2018deep,lugosch2017neural, vasic2018learning,nachmani2019hyper,doan2019neural,buchberger2020pruning,satorras2021neural,nachmani2022neural, gruber2017deep, shlezinger2019viterbinet, seo2018decoding, soltani2019deep,dai2020deep,eldar2022machine},  modulation classification \cite{erpek2020deep, peng2018modulation, liu2017deep}, and physical layer algorithms \cite{fritschek2019deep, o2017introduction, wang2017deep}.     
Within the context of channel decoding, prior works have demonstrated that deep neural network based decoders achieve lower bit/frame error rates than  conventional iterative decoding algorithms such as belief propagation in several  signal-to-noise ratio (SNR) regimes \cite{nachmani2016learning, nachmani2018deep, nachmani2019hyper, gruber2017deep, shlezinger2019viterbinet, seo2018decoding}.
In another line of works \cite{kim2018deepcode,jiang2019turbo,choi2019neural}, deep neural networks
have been used to jointly design \textit{both} encoder and decoder.  
Given the expansive applicability of deep neural networks for channel encoding and decoding, we note here that determining neural network architectures that generalize well to large block length codewords is an active area of research.

Iterative decoding algorithms (such as belief propagation (BP)) are commonly deployed for decoding linear codes; and are known to be equivalent to maximum aposteriori (MAP) decoding when the Tanner graph does not contain short cycles \cite{tang2005codes}. 
 However, if the Tanner graph contains short cycles, then BP can be sub-optimal i.e., the messages passed between the variable nodes and parity check nodes cannot correctly recover the transmitted codeword \cite{liu2012low, zhang2011causes, nachmani2016learning}.
  One approach to mitigate the effect of short cycles is by generalizing the BP algorithm  by means of a deep learning based approach \cite{nachmani2016learning,lugosch2017neural,nachmani2018deep,vasic2018learning,nachmani2019hyper,doan2019neural,buchberger2020pruning,satorras2021neural,nachmani2022neural}.
     It is shown that the weights learnt by optimizing over the training data ensure that any message repetition between the variable nodes and parity check nodes do not adversely impact the performance of BP based decoders \cite{nachmani2016learning, nachmani2018deep}.   
   We refer to this class of belief propagation decoders as \textit{Neural Belief Propagation} (NBP) decoders. 
   The salient aspect of NBP decoders is that its structure is determined from the corresponding Tanner graph, and therefore its architecture is a function of the code parameters itself. Several variants of NBP decoders have been a subject of recent study \cite{lugosch2017neural, vasic2018learning,nachmani2019hyper,doan2019neural, buchberger2020pruning, satorras2021neural, nachmani2022neural}. 
In \cite{lugosch2017neural}, the authors propose a hardware efficient implementation of the NBP decoder by reducing the number of matrix multiplications.
  The authors in \cite{nachmani2019hyper} implement message passing on graph neural networks wherein the output of each variable node is computed using a sub-network. 
  An interesting variant of NBP decoder was proposed in \cite{buchberger2020pruning} in which the unimportant check nodes were pruned in each decoding iteration thereby resulting in a architecture that corresponds to a different parity check matrix at each iteration. The authors in \cite{satorras2021neural} propose correcting the output of conventional decoding algorithms using a NBP decoder thereby combining the desirable features of both conventional and NBP decoders.
  A knowledge distillation based technique to learn the node activations in NBP decoder was proposed in \cite{nachmani2022neural}.
   
    Post-training, it is important that the NBP decoder achieves low bit-error-rate (BER) on unseen noisy codewords.
  Prior works on NBP decoders \cite{lugosch2017neural, vasic2018learning,nachmani2019hyper,doan2019neural, buchberger2020pruning, satorras2021neural, nachmani2022neural} are  empirical; 
to the best of our knowledge there are no theoretical guarantees on the performance of NBP decoders on unseen data.
To this end, given a NBP decoder, our goal is to understand how its architecture impacts its generalization gap  \cite{mohri2018foundations}, defined as the difference between empirical and expected BER(s).   
Motivated by the above discussion, we ask the following fundamental question: \textit{Given a NBP decoder, what is the expected performance on unseen noisy codewords? And how is the generalization gap related to code parameters, neural decoder architecture and training dataset size?} 

There are several approaches to obtaining generalization gap bounds in the theoretical machine learning literature, which can be classified into two primary categories.
The first category comprises data-independent approaches, such as VC-dimension, Rademacher complexity of the function class, and PAC-Bayes \cite{bartlett2019nearly,sontag1998vc,dziugaite2017computing,liao2020pac,mohri2008rademacher,mohri2018foundations,bartlett2002rademacher,bartlett2017spectrally,mansour2009domain,neyshabur2015norm}.
The second category focuses on data-dependent approaches, which analyze the mutual information between the input dataset and the algorithm output \cite{xu2017information,bu2020tightening}.
VC-dimension is a measure of the number of samples required to find a probably approximately correct (PAC) hypothesis from the entire hypothesis class \cite{kearns1994introduction, valiant1984theory, bousquet2021theory}.
Rademacher complexity measures the correlation between the function class and the random labels \cite{mohri2018foundations}; it is known that generalization bounds obtained via Rademacher complexity (and it's variants) are tighter than the bounds obtained using the VC-dimension approach \cite{mohri2008rademacher}.
Another method is the PAC-Bayes analysis, where generalization gap is bounded by the Kullback-Leibler divergence between the prior and the posterior on the learned weights.
The prior is chosen to be a multi-variate normal distribution centred around the initial weights \cite{langford2001not,neyshabur2017exploring, dziugaite2017computing,neyshabur2017pac}.
In  \cite{keskar2016large}, the authors propose the measuring the change in the training error with respect to perturbations in the model weights as a measure of its generalizability. 
Recent literature has increasingly focused on analyzing machine learning-based communication systems, particularly in terms of the generalization gap. 
For instance, \cite{weinberger2021generalization} investigates generalization bounds in the context of codebook design and decoder selection in both uncoded and coded communication systems. 
This paper specifically examines a setting where the encoder and decoder are learned in a data-driven manner. 
In uncoded systems, the authors consider minimum distance decoders, while in coded systems, they focus on learning decoders by maximizing the mutual information between the input and the channel output. 
The study highlights how the generalization gap scales with codebook size, the number of noise samples, and input dimensionality. 
Additionally, \cite{tsvieli2023learning} addresses the learning of decoders for predetermined codebooks. 
Drawing inspiration from the support vector machine paradigm, the authors propose a nearest neighbor decoder, the parameters of which are learned in a data-driven manner. 
The paper also derives bounds for the generalization gap of this learned decoder.
For the scope of this paper, we adopt the PAC learning framework and use Rademacher Complexity to understand the generalization gap of NBP decoders.  
 The notations introduced throughout the paper are summarized in Table \ref{tab: shorthand-notations-table-1-table-2}. We next summarize our main contributions.

  \begin{table*}[!t] 
\centering
\vspace{.1em}
\renewcommand{\arraystretch}{1.00}%
\vspace{.1em}
\begin{tabular}{p{.28\textwidth}p{.67\textwidth}}
\toprule
$\mathbf{y}[i] \rightarrow$ $i$-{th} index in  $\mathbf{y}$ & $\mathbf{H}[i, j] \rightarrow$   $i$-th row and $j$-th column entry in $\mathbf{H}$ \\
$n \rightarrow$ Blocklength &  $\mathcal{C} \rightarrow$ Linear block code of length $n$ and dimension $k$  \\
$k \rightarrow$  Dimension of the code & $f(\bm{\lambda}){[j]} \rightarrow$ $j$-th output bit of NBP decoder for input $\bm{\lambda}$ \\
$d_v \rightarrow$  Variable node degree &  $S = \{(\bm{\lambda}_j,\mathbf{x}_j)\}_{j=1}^{m} \rightarrow$  Dataset to train NBP decoder $f$  \\
$\mathbf{k} \rightarrow$  Message vector & $\mathbf{x}$, $\mathbf{y}$, $\mathbf{z} \rightarrow$  Channel input, output, and noise, respectively \\
$\kappa \rightarrow$ Code rate &  $\bm{\lambda}$, $\mathbf{\hat{x}} \rightarrow$  Decoder input, decoder output, respectively  \\
$T \rightarrow$ Decoding iterations &  $\mathcal{F}_T \rightarrow$ Function class of NBP decoders with $T$ iterations \\
$l_{\text{BER}}(\cdot) \rightarrow$ BER loss  &  $\mathcal{N}(\mathcal{F}_T ,\epsilon , \|\cdot\|_k) \rightarrow$ Covering number of $\mathcal{F}_T$ with respect to $k^{th}$ norm \\
$\mathcal{F}_{L,T} \rightarrow$ Hypothesis class & $\mathcal{M}(\mathcal{F}_T ,\epsilon , \|\cdot\|_k) \rightarrow$ Packing number of $\mathcal{F}_T$ with respect to $k^{th}$ norm\\
$\mathbf{H} \rightarrow$  Parity check matrix & $\mathcal{P}(\mathcal{F}_T ,\epsilon , \|\cdot\|_k) \rightarrow$ Packing of $\mathcal{F}_T$ with respect to the $k^{th}$ norm\\
$\mathcal{G} \rightarrow$ Tanner graph & $\mathcal{R}_{\text{BER}}(f) \rightarrow$ True risk of NBP decoder $f$ \\
$\mathcal{V} \rightarrow$ Variable nodes set  & $\mathcal{\hat{R}}_{\text{BER}}(f) \rightarrow$ Empirical risk of NBP decoder $f$  \\
$\mathcal{P} \rightarrow$ Parity check nodes set &  $R_{m} (\mathcal{F}_{L,T}) \rightarrow$  Empirical Rademacher complexity \\
$\mathcal{E} \rightarrow$ Edge set in graph $\mathcal{G}$ & $R_{m} (\mathcal{F}_{T}[j]) \rightarrow$  Bit-wise Rademacher complexity \\
$\mathbf{W}_i \rightarrow$ Weight matrices & $\{v_i,p_j\} \rightarrow$  Edge connecting variable node $v_i$, parity check node $p_j$ \\
$B \rightarrow$ Norm bounds & $\mathbf{v_t} \rightarrow$ Output of variable node hidden layer in $t$-{th} iteration \\
$b \rightarrow$ Absolute value bound & $\mathbf{p_t} \rightarrow$ Output of parity-check node hidden layer in $t$-{th} iteration \\
$b_{\lambda} \rightarrow$ Bit-wise bound on $\bm{\lambda}$ & $\mathbf{v_t}[\{l,m\}] \rightarrow$  Message from variable node $v_l$ to parity node $p_m$ \\
$w \rightarrow$ Bound on weights & $\mathbf{p_t}[\{l,m\}] \rightarrow$  Message from parity node $p_m$ to variable node $v_l$ \\
$||\cdot||_2 \rightarrow$ Spectral norm &  $\mathbf{W}_1[\{l,m\},l] \rightarrow$ Weight between  $l$-th input $\bm{\lambda}[l]$ and node $\mathbf{v_t}[\{l,m\}]$\\
$||\cdot||_{F} \rightarrow$ Frobenius norm  &  $\mathbf{W}_2[\{l,m\},\{l,m^{\prime}\}] \rightarrow$ Weight between nodes $\mathbf{v_t}[\{l,m\}]$,$\mathbf{p_t}[\{l,m^{\prime}\}]$ \\
$||\cdot||_{1} \rightarrow$ Max. column sum &  $\mathbf{W}_3[l,\{l,m\}] \rightarrow$ Weight between $\mathbf{v_T}[\{l,m\}]$ and $l$-th output $\mathbf{\hat{x}}[l]$\\
$||\cdot||_{\infty} \rightarrow$ Max. row sum & $\mathbf{W}_4[l,l] \rightarrow$ Weight between $\bm{\lambda}[l]$ and $\mathbf{\hat{x}}[l]$\\
$B_{W_i} \rightarrow$ Bound on $\|\mathbf{W}_i\|_2$ & 
$B_{w_i} \rightarrow$ $L_2$ norm bounds of rows in $\mathbf{W}_i$, where $i \in \{1,2,3,4\}$ \\
\bottomrule
\end{tabular}
\caption{Notations used in the paper.\label{tab: shorthand-notations-table-1-table-2}}
\end{table*}
   
\noindent \textbf{Main Contributions:} 
\begin{enumerate}
\item \textit{Generalization gap as a function of the covering number of the NBP decoder}:
In this paper, we first upper bound the generalization gap of a generic deep learning decoder as a function of the Rademacher complexity of the individual bits of the decoder output (which we denote as the bit-wise Rademacher complexity).
We next consider NBP decoders which belong to the class of belief propagation decoders whose architecture is a function of the code parameters.
 We upper bound the bit-wise Rademacher complexity as a function of the \textit{covering number} of the NBP decoder, which is the cardinality of the set of all decoders that can closely approximate the NBP decoder. 
  The covering number analysis provides an upper bound with a linear dependence of the generalization gap on spectral norm of the weight matrices and polynomial dependence on the decoding iterations. 
 The bound we obtain is tighter than the other approaches such as VC-dimension and PAC-Bayes approaches in which the upper bound exponentially depends on the decoding iterations.

\item \textit{Upper bounds on bit-wise Rademacher complexity for regular and irregular parity check matrices}:
We upper bound the covering number of NBP decoder in terms of the code-parameters (blocklength, variable node degree, check node degree), and the training dataset size for both regular and irregular parity check matrices.  
 From our results, we show that the generalization gap scales with the inverse of the square root of the dataset size, linearly with the variable node degree and the decoding iterations, and the square-root of the blocklength.
   To the best of our knowledge, this is the first result that determines upper bounds on the generalization gap as a function of the code-parameters.

\item \textit{Experimental evaluation of the generalization gap bounds}:
 We also present simulation results to validate our
theoretical findings. 
To the best of our knowledge, this is the first work that empirically studies the generalization gap of NBP decoders.
 In the experimental results, we consider binary phase shift keying (BPSK) modulation and additive white Gaussian noise (AWGN) channel. 
 We use Tanner code to illustrate the dependence of the generalization gap on the decoding iterations, and training dataset size. 
To study the dependence of the generalization gap on the blocklength, we consider two QC-LDPC parent codes, and generate descendent codes with smaller blocklengths by puncturing the parent code. 
 We assume that all-zero codewords are transmitted, and the NBP decoder is trained on the noisy realizations generated for a given channel signal-to-noise ratio (SNR). 
 In our empirical results, we observe that the generalization gap has a linear dependence on the decoding iterations, and it increases with the blocklength, thereby agreeing with the theoretically derived bounds. 
\end{enumerate}

\section{Preliminaries and Problem Statement}
\label{sec: preliminaries}
\begin{figure*}[!t]
	\begin{center}
		\includegraphics[scale = 0.28]{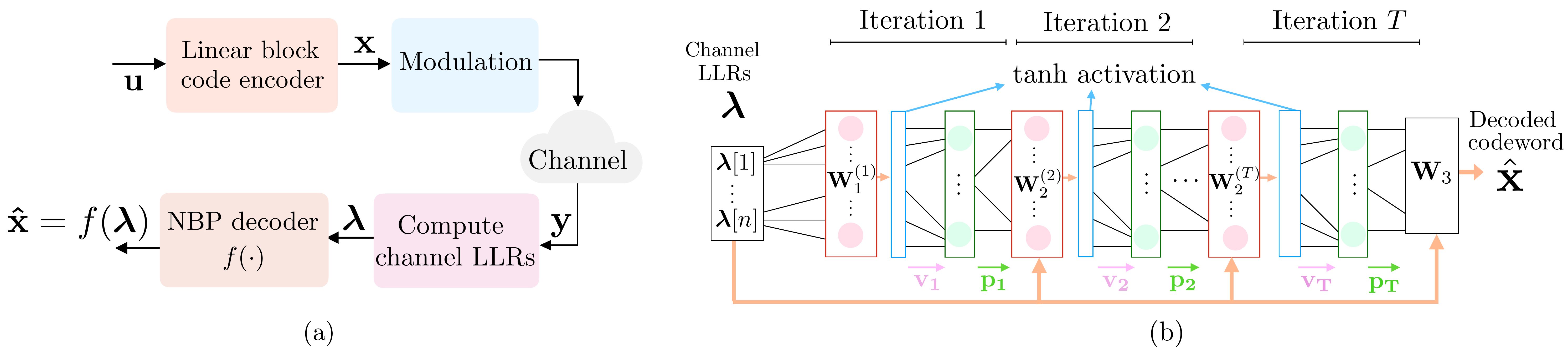}
	\caption{(a) End-to-End block diagram for communication using neural belief propagation (NBP) decoders for linear block codes; (b) Architecture of the NBP decoder for $T$ decoding iterations where each decoding iteration corresponds to $2$ hidden layers: (1) variable node layer, (2) parity check node layer. \label{fig: neural-network-decoder}}
 \vspace{-0.0cm}
		\end{center}
\end{figure*}
\noindent In Fig. \ref{fig: neural-network-decoder}, we consider a linear block code denoted by $\mathcal{C}$ of blocklength $n$  and message length $k$. 
Let the code $\mathcal{C}$ be characterized by a \textit{regular} parity check matrix $\mathbf{H} \in \{0,1\}^{ (n-k) \times n}$, and we denote the Tanner graph as $\mathcal{G} = \left(\mathcal{V} ,\mathcal{P},\mathcal{E}\right)$; where $\mathcal{V} = \{v_1,\cdots,v_n\}$ is the set of variable nodes, $\mathcal{P} = \{p_1,\cdots,p_{n-k}\}$ is the set of parity check nodes, and $\mathcal{E} = \{e_1,\cdots,e_{nd_v}\}$ is the set of edges.
Here, $d_v$ represents the variable node degree, i.e., the number of parity checks a variable node participates in.
Let $\{v_i,p_j\}$ denote the edge in the Tanner graph $\mathcal{G}$ connecting variable node $v_i$ to parity check node $p_j$. 
 $\mathcal{V} \left( v_j \right) = \{p_i|\mathbf{H}[i,j] = 1\}$ denote the set of parity check nodes adjacent to the variable node $v_j$ in the Tanner graph $\mathcal{G}$. 
 Similarly,  $\mathcal{P} \left( p_i \right) = \{v_j|\mathbf{H}[i,j] = 1\}$ denote the set of variable nodes adjacent to the parity check node $p_i$ in $\mathcal{G}$.

 Let $\mathcal{Y} \subseteq \mathbb{R}^{n}$ be the space of $n$ dimensional channel outputs, $\mathcal{X} \subseteq \{0,1\}^{n} $ be the space of $n$ dimensional codewords, $\mathcal{U} \subseteq \{0,1\}^{k}$ be the space of $k$ dimensional messages, and $\mathcal{Z} \subseteq \mathbb{R}^{n}$ be the space of $n$ dimensional channel noise.
The message $\mathbf{u} = [\mathbf{u}[1], \cdots, \mathbf{u}[k]]^{\top}  \in \mathcal{U}$ is encoded to the codeword $\mathbf{x} = [\mathbf{x}[1], \cdots, \mathbf{x}[n]]^{\top} \in \mathcal{X}$.
The channel is assumed to be memoryless, described by $\Pr(\mathbf{y}|\mathbf{x}) = \prod_{i = 1}^n \Pr(\mathbf{y}[i]|\mathbf{x}[i])$. 
The receiver receives the channel output $\mathbf{y} = [\mathbf{y}[1], \cdots, \mathbf{y}[n]]^{\top} \in \mathcal{Y}$; 
which is the codeword $\mathbf{x}$ modulated, and corrupted with  independently and identically distributed (i.i.d.) additive noise $\mathbf{z} = [\mathbf{z}[1], \cdots, \mathbf{z}[n]]^{\top} \in \mathcal{Z}$. 

The goal of the decoder is to recover the message $\mathbf{u}$ from the channel output $\mathbf{y}$. 
The input to the decoder is the log-likelihood ratio (LLR) of the posterior probabilities denoted by $\bm{\lambda} \in \mathbb{R}^{n \times 1}$ and is given as $ \bm{\lambda}[i] = \log \frac{\Pr(\mathbf{x}[i]=0|\mathbf{y}[i])}{\Pr(\mathbf{x}[i]=1|\mathbf{y}[i])} $, for $1 \leq i \leq n$.
Denote the output of the NBP decoder with $T$ decoding iterations as $\mathbf{\hat{x}} = f(\bm{\lambda})$, where $f(\cdot)$ denotes the decoding function.

The architecture of the NBP decoder is derived from the trellis representation of $\mathcal{G}$ and illustrated in Fig. \ref{fig: neural-network-decoder}(b). 
Each decoding iteration $t$ (where, $1 \leq t \leq T$) corresponds to two hidden layers each of width $|\mathcal{E}| = n d_v$, namely: (1) variable layer $\mathbf{v_t}$, (2) parity check layer $\mathbf{p_t}$.  
The hidden nodes in layers $\mathbf{v_t}$ and $\mathbf{p_t}$ correspond to the messages passed along the edges of the Tanner graph $\mathcal{G}$.
For instance, the output of the node $\mathbf{v_t}[\{l,m\}]$ in the NBP decoder corresponds to the message passed from variable node $v_l$ to parity check node $p_m$ in the $t$-{th} iteration, and is given as,

\begin{align}
    \mathbf{v_t}[\{l,m\}] =  \mathbf{W}^{(t)}_1[\{l,m\},l]\bm{\lambda}[l] +   \underset{m^{\prime} \in \mathcal{V}(l)\backslash m}{\sum} \mathbf{W}^{(t)}_2[\{l,m\},\{l,m^{\prime}\}]\mathbf{p_{t-1}}[\{l,m^{\prime}\}],
\end{align}

where, $\mathbf{p_{t-1}}[\{l,m^{\prime}\}]$ corresponds to the message passed from the parity check node $p_{m^{\prime}}$ to the variable node $v_l$ in the $(t-1)$-{th} iteration.
For $t = 1$, we have $\mathbf{p_{0}} = \left[0,\cdots,0 \right]^{\top}$.
$\mathbf{W}^{(t)}_1 \in \mathbb{R}^{nd_v \times n}$, and $\mathbf{W}^{(t)}_2 \in \mathbb{R}^{nd_v \times nd_v}$ are sparse weight matrices trained using backpropagation in the $t$-{th} decoding iteration. 
$\mathbf{W}^{(t)}_1$ is strictly a lower triangular matrix with exactly $d_v$ non-zero entries in every column, and one non-zero entry in every row.
$\mathbf{W}^{(t)}_2$ has exactly $d_v-1$ non-zero entries in every row, and $d_v-1$ non-zero entries in every column. 
We consider that the $t$-{th} decoding iteration is characterized by weight matrices $\mathbf{W}^{(t)}_1$, and $\mathbf{W}^{(t)}_2$, where $t$ can take integer values $t \in \{1,\cdots, T\}$.
The output of the parity check node layer in the $t$-{th} decoding iteration for the NBP decoder is,

\begin{align}
\mathbf{p_t}[\{l,m\}]=2 \tanh ^{-1}\left(\prod_{l^{\prime} \in \mathcal{P}(m)\backslash l} \tanh \left( \frac{{\mathbf{v_t}[\{l^{\prime} ,m\}]}}{2} \right)\right)
\label{eq: belief-propagation-without-approximation}
\end{align}
Implementing \eqref{eq: belief-propagation-without-approximation} is computationally expensive in hardware due to the multiplicative operations and hyperbolic functions. 
For practical implementation, \eqref{eq: belief-propagation-without-approximation} can be made computationally feasible by using the min-sum operation, which is described as follows:
 \begin{align}
   \hspace{-0.25cm} \mathbf{p_t}[\{l,m\}] = \hspace{-0.4cm} \underset{l^{\prime} \in \mathcal{P}(m)\backslash l}{\prod} \hspace{-0.35cm} sign({\mathbf{v_t}[\{l^{\prime} ,m\}]}) \underset{l^{\prime} \in \mathcal{P}(m)\backslash l}{\min} |{\mathbf{v_t}[\{l^{\prime} ,m\}]}|.
    \label{eq: min-sum-approximation}
\end{align}
We note that learnable parameters can be incorporated into the min-sum operations.
Specifically, the output of the parity check node layer can be scaled with weights as follows:
 \begin{align}
   \hspace{-0.25cm} \mathbf{p_t}[\{l,m\}] =  \bm{\beta_t}[\{l,m\}] 
 \underset{l^{\prime} \in \mathcal{P}(m)\backslash l}{\prod} sign({\mathbf{v_t}[\{l^{\prime} ,m\}]}) \hspace{0.1cm} \mathbf{\tilde{p}_t}[\{l,m\}].
 \vspace{-0.25cm}
\end{align}
The parameter vector $\bm{\beta_t}$ is learned in a data-driven manner. Alternatively, the output of the parity check node layer can be offset using the parameter $\bm{\beta_t}$ as follows: 
 \begin{align}
   \hspace{-0.25cm} \mathbf{p_t}[\{l,m\}] = \hspace{-0.4cm} \underset{l^{\prime} \in \mathcal{P}(m)\backslash l}{\prod} sign({\mathbf{v_t}[\{l^{\prime} ,m\}]}) \hspace{0.1cm} ReLu\left(\mathbf{\tilde{p}_t}[\{l,m\}] - \bm{\beta_t}[\{l,m\}]  \right).
   \vspace{-0.25cm}
\end{align}

\noindent In this paper, we concentrate on the scenario where the learnable parameters are used solely for computing the output of the variable node layer.
The estimated codeword after $T$ decoding iterations in the NBP decoder is given as, 
\begin{align}
   \hspace{-0.3cm} \mathbf{\hat{x}}[l] \hspace{-0.05cm} = \hspace{-0.05cm} s( \mathbf{W}^{(T)}_4[l,l] \bm{\lambda}[l] + \hspace{-0.3cm} \underset{m^{\prime} \in \mathcal{V}(l)}{\sum} \hspace{-0.22cm} \mathbf{W}^{(T)}_3[l,\{l,m^{\prime}\}] \mathbf{p_T}[\{l,m^{\prime}\}])
\end{align}
where,  $\mathbf{W}_3 \in \mathbb{R}^{n \times nd_v}$, $\mathbf{W}_4 \in \mathbb{R}^{n \times n}$, and $s(\cdot)$ is the sigmoid activation. 
$\mathbf{W}_3$ is strictly an upper triangular matrix with exactly $d_v$ non-zero entries in every row, while $\mathbf{W}_4$ is a diagonal matrix. 

\noindent The NBP decoder (denoted by $f(\cdot)$) is characterized by the following four sparse weight matrices: (a) $\mathbf{W}^{(t)}_1$, where $t = 1, \cdots, T$, (b) $\mathbf{W}^{(t)}_2$, where $t =  1, \cdots, T$, (c) $\mathbf{W}_3$, and (d) $\mathbf{W}_4$.
The weight matrices are learnt by training the NBP decoder to minimize the bit error rate (BER) loss that is defined as, 
\begin{align}
    l_\text{BER}(f(\bm{\lambda}),\mathbf{x})=\frac{d_H(f(\bm{\lambda}),\mathbf{x})}{n} = \frac{\sum_{j = 1}^{n} \mathbbm{1}\left(f({\bm{\lambda}}){[j]} \neq \mathbf{x}[j] \right)}{n}.
\end{align}
Here, $d_H(\cdot,\cdot)$ denotes the Hamming distance, and $\mathbbm{1}(\cdot)$ denotes the indicator function.
In practice, we train the NBP decoder to minimize the BER loss over the  dataset  $S = \{(\bm{\lambda}_j,\mathbf{x}_j)\}_{j=1}^{m}$ comprising of pairs of log-likelihood ratio and its corresponding codeword.
Then, we define the empirical risk of $f$ as $\mathcal{\hat{R}}_{\text{BER}}(f) = \frac{1}{m}\sum\limits_{j = 1}^m  l_\text{BER}(f(\bm{\lambda}_j),\mathbf{x}_j)$.  The true risk of $f$ is defined as $\mathcal{R}_{\text{BER}}(f) = \mathop{\mathbb{E}}_{{\bm{\lambda}},{\mathbf{x}}}[l_\text{BER}(f(\bm{\lambda}),\mathbf{x})]$. 

\noindent \textbf{Problem Statement.}  
The generalization gap is defined as the difference $\mathcal{R}_{\text{BER}}(f)-\mathcal{\hat{R}}_{\text{BER}}(f) $. The main goal of this paper is to understand the behavior of the generalization gap (specifically upper bounds) as a function of a) training dataset size, $m$,  b) the \textit{complexity} of the NBP decoder, in terms of the number of decoding iterations $T$ and c) code parameters, such as message length $k$, blocklength $n$, variable node degree $d_v$, parity check node degree $d_c$.

\section{Main Results}
\label{sec: main-results}
\noindent In this section, we present our main results on the generalization gap for NBP decoders.
Let $S = \{(\bm{\lambda}_j,\mathbf{x}_j)\}_{j=1}^{m}$ be the training dataset, and we assume that the dataset is i.i.d from a fixed distribution.
Let $\mathcal{F}_T$ be a class of NBP decoders with $T$ decoding iterations. For the scope of this paper, we focus on the family of NBP decoders whose non-zero weight entries are bounded by a constant $w$. Specifically, we assume that for every $(i,j)$ and $1\leq t \leq T$, $|\mathbf{W}^{(t)}_1[i,j]|\leq w$, $|\mathbf{W}^{(t)}_2[i,j]|\leq w$, $|\mathbf{W}_3[i,j]|\leq w$ and $|\mathbf{W}_4[i,j]|\leq w$, i.e., the maximum absolute value of the $(i,j)$ coordinates for all the weight matrices are bounded by a non-negative constant $w$. In addition, we also assume that input log-likelihood ratio $|\bm{\lambda}[i]| \leq b_{\lambda}$ for all $i=1,\ldots, n$.  

We define the hypothesis class $\mathcal{F}_{L,T}$, derived from the class $\mathcal{F}_T$ of NBP decoders as follows:  
\begin{align}
    \mathcal{F}_{L,T} = \left\{({\bm{\lambda}},\mathbf{x})\mapsto l_\text{BER}(f({\bm{\lambda}}),\mathbf{x}):f\in \mathcal{F}_T\right\}.
\label{eq: definition set FLT}    
\end{align}
Intuitively, for each $f \in \mathcal{F}_T$, the output of the corresponding function in $\mathcal{F}_{L,T}$ is the BER loss of the decoder $f$. We next define the empirical Rademacher complexity of $\mathcal{F}_{L,T}$.
\begin{definition} \textit{(Rademacher complexity of $\mathcal{F}_{L,T}$)} 
The empirical Rademacher complexity of $\mathcal{F}_{L,T}$ is defined as
\begin{align}
    R_{m} (\mathcal{F}_{L,T}) \triangleq \underset{\sigma}{\mathbb{E}} \left[\underset{_{f \in \mathcal{F}_{T}}}{\sup} \frac{1}{m}\sum_{i = 1}^{m}\sigma_{i} l_\text{BER}(f({\bm{\lambda}_i}),\mathbf{x}_i)\right], 
\end{align}
where $\sigma_{i}$'s are i.i.d. Rademacher random variables, i.e., $\Pr(\sigma_i = 1) = \Pr(\sigma_i = -1) = \frac{1}{2}$. 
\label{def-1:empirical-Rademacher-complexity}
\end{definition}
 We note that the loss function $l_{\text{BER}}$ takes the values between $[0,1]$; and consequently using a standard result from PAC learning literature (Theorem 3.3 in \cite{mohri2018foundations}), one can bound the generalization gap in terms of $R_{m} (\mathcal{F}_{L,T})$. 
  Specifically, for any $\delta \in (0,1)$, with probability at least $1-\delta$, the generalization gap for any $f \in F_{T}$ is bounded as follows:
\begin{align}
    \mathcal{R}_\text{BER}(f)
    - \mathcal{\hat{R}}_{\text{BER}}(f) \leq  2 R_m(\mathcal{F}_{L,T})+\sqrt{\frac{\log(1/\delta)}{2m}}.
    \label{eq:Thm-1(a)}
\end{align}

To proceed further, we introduce bit-wise Rademacher complexity of $\mathcal{F}_T$; which is a new notion and captures the correlation between  $j$-th channel output of the NBP decoder and a random decision (Rademacher random variable). 
 
 \begin{definition} \textit{(Bit-wise Rademacher complexity of $\mathcal{F}_T$)} For a NBP decoder class $\mathcal{F}_T$, the empirical bit-wise Rademacher complexity corresponding to its $j$-th output bit is defined as:
\begin{align}
    R_m(\mathcal{F}_{T}[j]) \triangleq \underset{\sigma}{\mathbb{E}} \left[\underset{{f \in \mathcal{F}_{T}}}{\sup}\frac{1}{m}\sum_{i=1}^m\sigma_i \cdot f({\bm{\lambda}}_i){[j]}\right].
\end{align}
\label{def-2:bit-wise-Rademacher-complexity}
\end{definition}

\noindent We next present Proposition \ref{Proposition-1: bit-wise-Rademacher} in which we upper bound the generalization gap as a function of the empirical bit-wise Rademacher complexity $ R_m(\mathcal{F}_{T}[j])$.

\begin{proposition}
For any $\delta \in (0,1)$, with probability at least $1-\delta$, the generalization gap for any NBP decoder $f\in \mathcal{F}_T$ can be upper bounded as follows,
\begin{align}
   \hspace{-0.3cm} \mathcal{R}_\text{BER}(f)
    - \mathcal{\hat{R}}_{\text{BER}}(f) \leq \frac{1}{n} \sum_{j = 1}^{n} R_m(\mathcal{F}_{T}[j]) +\sqrt{\frac{\log(1/\delta)}{2m}},
    \label{eq:prop-1}
\end{align}
where $R_m(\mathcal{F}_{T}[j])$ denotes the bit-wise Rademacher complexity for the $j$th output bit. 
\label{Proposition-1: bit-wise-Rademacher}
\end{proposition}
The proof of Proposition \ref{Proposition-1: bit-wise-Rademacher} is presented in Appendix \ref{sec: Proposition-1-Proof}. We now present Theorem \ref{Theorem-1: dudley-entropy-intergral} which is the main result of this paper. The main technical challenge is to bound the bit-wise Rademacher complexity  $R_m(\mathcal{F}_{T}[j])$ as a function of  the number of decoding iterations $T$, training dataset size $m$ and code parameters (blocklength $n$ and variable node degree $d_v$).
\begin{theorem}
For any $\delta \in (0,1)$, with probability at least $1-\delta$, the generalization gap for any NBP decoder $f\in \mathcal{F}_T$ can be upper bounded as follows,
\begin{align}
    &\mathcal{R}_\text{BER}(f) - \mathcal{\hat{R}}_{\text{BER}}(f)
    \leq   \frac{4}{{m}} +\sqrt{\frac{\log(1/\delta)}{2m}} +  {12} \sqrt{\frac{(n d^2_v T + 1)({T+1})}{m} \log \left(8{\sqrt{mn}wd_vb_{\lambda}}\right)},
\end{align}
where $n$ denotes the blocklength, $d_v$ is the variable node degree, $T$ is the number of decoding iterations (number of layers in NBP), $m$ is the training dataset size;  $w$ and $b_{\lambda}$ are upper bounds on the weights in the NBP decoder and input log-likelihood ratio, respectively.
\label{Theorem-1: dudley-entropy-intergral}
\end{theorem}
\noindent \textbf{Proof-sketch of Theorem \ref{Theorem-1: dudley-entropy-intergral}:} The detailed  proof of Theorem \ref{Theorem-1: dudley-entropy-intergral} is presented in Appendix \ref{sec: Theorem-1-Proof} and here we briefly describe the main ideas. 
We first upper bound the bit-wise Rademacher complexity in terms of Dudley entropy integral (specifically, leveraging Massart's Lemma in \cite{bartlett2017spectrally} and adapting it to our problem). The resulting bound is expressed in terms of the covering number of the NBP decoder class, i.e., the \textit{smallest} cardinality of the set of functions in $\mathcal{F}_{T}$ that can closely approximate the NBP decoding function $f$.  To further bound the covering number, we first show that the NBP decoder is Lipschitz in its weight matrices which is proved in Lemma \ref{lemma-1: f is lipschitz} (see Appendix \ref{sec: Theorem-1-Proof}). In other words, for a given input, the output of the NBP decoder remains invariant to small perturbations in its weight matrices. Using this fact, we obtain a bound on the covering number of the NBP decoder class in terms of a product of covering numbers (each corresponding to a weight matrix). We then observe that the weight matrices for the NBP decoder are sparse, where the structure and number of non-zero entries is determined by the parity check matrix and the code parameters (such as blocklength $n$, variable node degree $d_v$ etc.).  We then use the fact that the covering number of a sparse weight matrix is always smaller than that of a non-sparse vector (of the same size as the total non-zero entries in the original sparse matrix). Using our result in Lemma \ref{Lemma-2: Covering number matrices}, we can finally upper bound the bit-wise Rademacher complexity as a function of the code parameters to deduce the result in Theorem \ref{Theorem-1: dudley-entropy-intergral}.

\begin{remark} [\textbf{Representation in Terms of Code-rate and Parity Check Node Degree}] 
The result in Theorem \ref{Theorem-1: dudley-entropy-intergral} can also be expressed as follows:
\begin{align}
  \hspace{-0.2cm} \mathcal{R}_\text{BER}(f) - \mathcal{\hat{R}}_{\text{BER}}(f)
    \leq  \frac{4}{{m}}  \hspace{-2pt} + \hspace{-2pt} \sqrt{\frac{\log(1/\delta)}{2m}} \hspace{-2pt} + \hspace{-2pt}
     {12} \sqrt{\frac{(n d^2_c (1 - \kappa)^2 T + 1)({T+1})}{m} \log \left(8{\sqrt{mn}wd_vb_{\lambda}}\right)}. 
    \label{dudley-entropy-intergral-b}
\end{align}
 We use the fact that the blocklength, message length, variable node degree, and parity check node degree are related as $nd_v = (n-k)d_c$.
 Using this relation in Theorem \ref{Theorem-1: dudley-entropy-intergral} we obtain \eqref{dudley-entropy-intergral-b}.
 From the result in \eqref{dudley-entropy-intergral-b} we note that the generalization gap reduces for codes with a high code-rate $\kappa$.
\end{remark}

\begin{figure}[!t]
	\begin{center}
		\includegraphics[scale = 0.25]{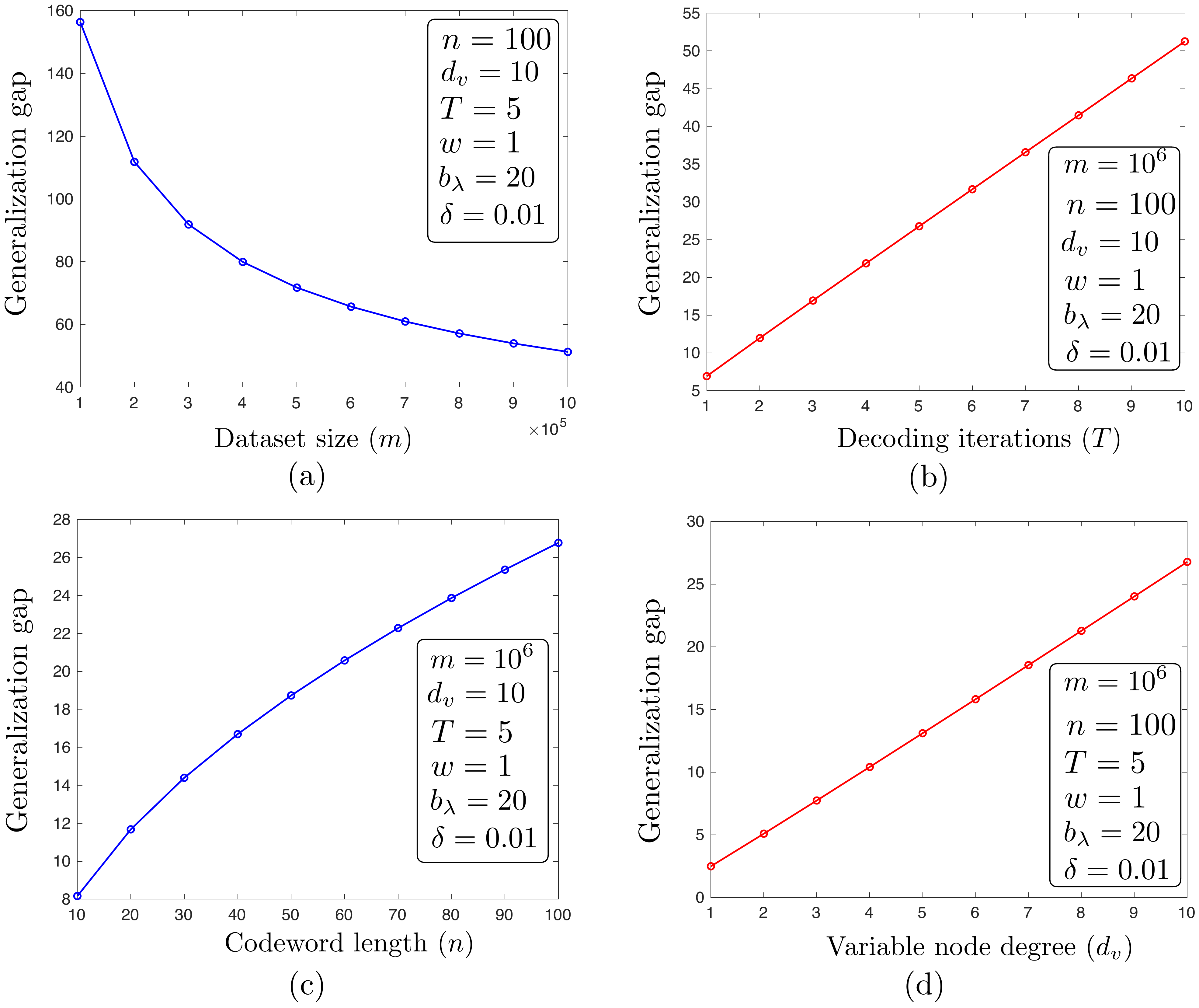}
	\caption{(a) RHS in Theorem \ref{Theorem-1: dudley-entropy-intergral} vs Dataset size ($m$), (b) RHS in Theorem \ref{Theorem-1: dudley-entropy-intergral} vs Decoding iterations ($T$), (c) RHS in Theorem \ref{Theorem-1: dudley-entropy-intergral} vs Blocklength ($n$), (d) RHS in Theorem \ref{Theorem-1: dudley-entropy-intergral} vs Variable node degree ($d_v$). \label{fig: gen_vs_parameters}}
		\end{center}
  \vspace{-20pt}
\end{figure}

\begin{remark} [\textbf{Impact of the Code-parameters}]  
 We plot the generalization gap bound obtained in Theorem \ref{Theorem-1: dudley-entropy-intergral} in  Fig. \ref{fig: gen_vs_parameters} for blocklength $n = 100$, variable node degree  $d_v = 10$, decoding iterations $T = 10$, and  dataset size $m = 10^6$.
 To understand the dependence of the generalization gap on a parameter, we vary that parameter while keeping the values of the remaining parameters fixed. 
Smaller training dataset size results in overfitting, and therefore corresponds to a larger generalization gap. 
 We observe this in Fig.  \ref{fig: gen_vs_parameters}(a), wherein the generalization gap decays as $\mathcal{O}(\frac{1}{\sqrt{m}})$. 
 While more decoding iterations (i.e.,  more hidden layers) are expected to improve decoding performance, it can also overfit the training data. 
 Therefore, we expect the generalization gap to increase with the number of decoding iterations. 
As seen from  Fig. \ref{fig: gen_vs_parameters}(b), we note that the generalization gap of the NBP decoder scales linearly as   $\mathcal{O}(T)$.
Our theoretical result in Theorem \ref{Theorem-1: dudley-entropy-intergral} tells us that the generalization gap scales with the blocklength as $\mathcal{O}(\sqrt{n})$ as shown in Fig. \ref{fig: gen_vs_parameters}(c). 
However, the generalization gap scales linearly with the variable node degree as $\mathcal{O}(d_v)$ as shown in Fig. \ref{fig: gen_vs_parameters}(d).

\end{remark}


\begin{remark} [\textbf{Comparison with Other Approaches for Bounding the Generalization Gap}] 
Vapnik-Chervonenkis (VC) dimension bounds \cite{bartlett2019nearly, sontag1998vc}, PAC-Bayes analysis \cite{dziugaite2017computing,neyshabur2017pac,liao2020pac} are other techniques to upper bound the generalization gap. 
While VC-dimension approach yields a bound independent of the data distribution, it is found that these bounds are vacuous \cite{dziugaite2017computing, dziugaite2020search} and scales exponentially with the number of parameters of the neural network. 
To obtain tighter and non vacuous generalization bounds, prior works \cite{biggs2022non, alquier2021user, dziugaite2017computing} have proposed the use of PAC-Bayes analysis.
For any $\delta \in (0,1)$, with probability at least $1-\delta$, the generalization gap using PAC-Bayes analysis is upper bounded as,
$\mathcal{R}_\text{BER}(f)
- \mathcal{\hat{R}}_{\text{BER}}(f) \leq  \sqrt{\frac{\textrm{KL}(\zeta||\Gamma)+ \log{\sqrt{m}}+ \log\left(2/\delta\right)}{2m} }.
$
The PAC-Bayes prior on the space of neural network decoders $\zeta$ is chosen independent of the training data \cite{biggs2022non, dziugaite2018data}.
The KL divergence term between the PAC-Bayes prior $\zeta$ and  posterior $\Gamma$ is typically the dominant term in the bound for the generalization gap. 
While the posterior  $\Gamma$ achieves minimal empirical risk, and is data-dependent; the KL divergence term can be large as the  data-independent priors are chosen arbitrarily causing the bound to be vacuous \cite{dziugaite2018data}.
In \cite{neyshabur2017pac}, the PAC-Bayes framework is utilized to establish an upper bound on the generalization gap as a function of the network's sharpness, where sharpness is defined as the change in network output relative to the perturbation of the weight matrices. 
The resulting bound is a function of the spectral and Frobenius norms of the weight matrices, assuming that both the perturbation and the prior are from a Gaussian distribution. 
Exploring PAC-Bayes analysis for the NBP (Neural Belief Propagation) decoder paradigm, and understanding how the bound scales with different priors, is an interesting future direction.
PAC-Learning approach used in this paper leads to a cleaner analysis (inspired by recent results on generalization bounds for graph neural networks and recurrent neural networks \cite{garg2020generalization, chen2019generalization}), and the bound obtained has a closed-form expression with explicit dependence on code parameters, decoding iterations, and the training dataset size.
\end{remark}

\noindent We next show that Theorem \ref{Theorem-1: dudley-entropy-intergral} can be readily generalized to irregular parity check matrices.  Specifically, consider an irregular parity check matrix $\mathbf{H} \in \{0,1\}^{(n-k) \times n}$  where $d_{v_i}$ is the variable node degree of the $i$-{th} bit in the codeword, and $d_{c_j}$ is the parity check node degree of the $j$-{th} parity check equation.
The NBP decoder corresponding to such this irregular parity check matrix is characterized by the weight matrices $\{\mathbf{W}^{(t)}_1| 1 \leq t \leq T\}$, $\{\mathbf{W}^{(t)}_2| 1 \leq t \leq T\}$, $\mathbf{W}_3$, $\mathbf{W}_4$. 
 Here, for every $1 \leq t \leq T$, and $\theta = \sum\limits_{i = 1}^{n} d_{v_i}$, we have that $\mathbf{W}^{(t)}_1 \in \mathbb{R}^{ \theta \times n}$, $\mathbf{W}^{(t)}_2 \in \mathbb{R}^{ \theta \times \theta}$, $\mathbf{W}_3 \in \mathbb{R}^{ n \times \theta}$, and 
$\mathbf{W}_4 \in \mathbb{R}^{ n \times n}$.
For any value of $t$, the weight matrix $\mathbf{W}^{(t)}_1$ has one non-zero entry in every row, and $d_{v_i}$ non-zero entries in the $i$-th column. 
In the weight matrix $\mathbf{W}^{(t)}_2$, the $i$-th bit in the codeword with variable node degree $d_{v_i}$ corresponds to $d_{v_i}$ rows and $d_{v_i}$ columns, and these rows and columns each have exactly $d_{v_{i}}-1$ non-zero entries.

\begin{theorem}
For any $\delta \in (0,1)$, with probability at least $1-\delta$, the generalization gap for any NBP decoder $f\in \mathcal{F}_T$ corresponding to irregular parity check matrix can be upper bounded as follows, 
\begin{align}
    \mathcal{R}_\text{BER}(f) -  \mathcal{\hat{R}}_{\text{BER}}(f)
    \leq \frac{4}{{m}} +\sqrt{\frac{\log(1/\delta)}{2m}} + {12} \sqrt{\frac{\sum\limits_{j = 1}^{n} d^2_{v_j}({T+1})^2}{m} \log \left(8{\sqrt{mn}w\max_{i} d_{v_i} b_{\lambda}}\right)}.
\end{align}
\label{Corollary-1: irregular-parity-check-mat}
\end{theorem}
\noindent The proof of Theorem \ref{Corollary-1: irregular-parity-check-mat} follows similar steps used to prove Theorem \ref{Theorem-1: dudley-entropy-intergral}, and is presented in Appendix \ref{sec: Corollary-1-Proof}.

\noindent In Theorem \ref{Theorem-1: dudley-entropy-intergral}, we assumed that the log-likelihood ratios are bounded and this result does not take  the channel SNR into account.
We study the impact of the bound on input log-likelihood ratios and the channel SNR in Theorem \ref{Proposition-2: unbounded-input-channel-snr} which is presented next.

\begin{theorem}
\label{Proposition-2: unbounded-input-channel-snr}
For any $\delta \in (0,1)$, with probability at least $1-\delta$, the generalization gap for any NBP decoder $f\in \mathcal{F}_T$ with unbounded log-likelihood ratios is upper bounded as follows,
 \begin{align}
   \mathcal{R}_\text{BER}(f) - \mathcal{\hat{R}}_{\text{BER}}(f)
    \leq  \min_{b_{\lambda}} \phi(n,d_v,T,m,w,b_{\lambda})
     + \frac{4}{{m}} +\sqrt{\frac{\log(1/\delta)}{2m}}.
     \label{eq: unbounded-channel-snr-prop2}
 \end{align}
 where, $\phi(n,d_v,T,m,w,b_{\lambda}) = {12} \sqrt{\frac{(n d^2_v T + 1)({T+1})}{m} \log \left(8{\sqrt{mn}wd_vb_{\lambda}}\right)} + \operatorname{Pr}\left(\exists i \in [n] \text{,s.t.} |\bm{\lambda}[i]| > b_{\lambda} \right) $.
 Suppose the symbols are modulated using binary phase shift keying (BPSK) modulation, and the channel is AWGN with variance $\beta^2$, then $\operatorname{Pr}\left(\exists i \in [n] \text{ s.t. } |\bm{\lambda}[i]| > b_{\lambda} \right) = \left(1 - Q\left(\frac{\beta^2 b_{\lambda}+ 2}{2\beta} \right) - Q\left(\frac{\beta^2 b_{\lambda}- 2}{2\beta} \right) \right)^n$.
 \label{Proposition-2: unbounded-channel-snr}
\end{theorem}

\noindent The proof of Theorem \ref{Proposition-2: unbounded-channel-snr} is presented in Appendix \ref{sec: Proposition-2-Proof}.
To take the unbounded input log-likelihood ratio into account for analysis, we use the law of total expectations, and condition the true risk with the event that the input log-likelihood ratio is not bounded, i.e., $|\bm{\lambda}[i]| > b_{\lambda}$ for any $i \in [n]$. 
We bound the probability that log-likelihood ratio is unbounded assuming BPSK modulation, and AWGN channel.  
In addition, we have the true risk conditioned on the event that the input log-likelihood ratio is bounded which directly follows from Theorem \ref{Theorem-1: dudley-entropy-intergral} in this paper. 
In \eqref{eq: unbounded-channel-snr-prop2}, the term ${12} \sqrt{\frac{(n d^2_v T + 1)({T+1})}{m} \log \left(8{\sqrt{mn}wd_vb_{\lambda}}\right)}$ is an increasing function of $b_{\lambda}$, and $\operatorname{Pr}\left(\exists i \in [n] \text{,s.t.} |\bm{\lambda}[i]| > b_{\lambda} \right) $ is a decreasing function of $b_{\lambda}$.
Therefore, minimizing the two terms over $b_{\lambda}$ provides the upper bound on the generalization gap.

\begin{remark} [\textbf{Minimizing the generalization gap by selecting the bound on LLR ($b_{\lambda}$) based on Channel SNR}]
The generalization gap in Theorem \ref{Proposition-2: unbounded-channel-snr} comprises of two terms: (a) ${12} \sqrt{\frac{(n d^2_v T + 1)({T+1})}{m} \log \left(8{\sqrt{mn}wd_vb_{\lambda}}\right)}$, which increases with $b_{\lambda}$; (b) $\operatorname{Pr}\left(\exists i \in [n] \text{,s.t.} |\bm{\lambda}[i]| > b_{\lambda} \right) $, which is a decreasing function of $b_{\lambda}$, $\beta$.
The choice of $b_{\lambda}$ to minimize the total generalization gap is a function of $\beta$.
To illustrate this, we plot the generalization gap bound obtained in Theorem \ref{Proposition-2: unbounded-input-channel-snr}, the generalization gap bound obtained in Theorem \ref{Theorem-1: dudley-entropy-intergral} for bounded input LLR, and the probability that the input LLR is unbounded in Fig. \ref{fig: Theorem-3-remark}(a). 
We set blocklength $n = 100$, variable node degree $d_v = 10$, decoding iterations $T = 10$, and dataset size $m = 10^6$.
As seen in Fig.  \ref{fig: Theorem-3-remark}(a), the generalization gap terms obtained in Theorem \ref{Theorem-1: dudley-entropy-intergral} are minimized in region R1 that corresponds to smaller values of $\beta$ (or large channel SNR).
This behavior is attributed to lower values of  $b_{\lambda}$ (i.e., $b_{\lambda} = 10$) as observed in Fig. \ref{fig: Theorem-3-remark}(b).
As the $\beta$ is increased (or reducing the channel SNR) in region R2 in Fig. \ref{fig: Theorem-3-remark}(a), the minimum generalization gap is obtained for larger values of $b_{\lambda}$.
The term $\operatorname{Pr}\left(\exists i \in [n] \text{,s.t.} |\bm{\lambda}[i]| > b_{\lambda} \right)$ also decreases with increase in $\beta$ (or lower channel SNR values). 
In other words, there is a trade-off between the terms ${12} \sqrt{\frac{(n d^2_v T + 1)({T+1})}{m} \log \left(8{\sqrt{mn}wd_vb_{\lambda}}\right)}$, and $\operatorname{Pr}\left(\exists i \in [n] \text{,s.t.} |\bm{\lambda}[i]| > b_{\lambda} \right) $ based on the channel SNR.
We adopt this approach to establish a dependency on the channel SNR, considering that our method for bounding the generalization gap was previously data-independent. In contrast, data-dependent bounds create a direct link between the generalization gap and the mutual information involving the input dataset and the algorithm output \cite{xu2017information,bu2020tightening}. This methodology implicitly accounts for various factors, including the dataset, hypothesis set, learning algorithm, and the loss function employed. Consequently, the integration of mutual information-based approaches could be crucial in demonstrating a direct correlation between the generalization gap and the channel SNR.
\end{remark}
\begin{figure}[!t]
	\begin{center}
		\includegraphics[scale = 0.32]{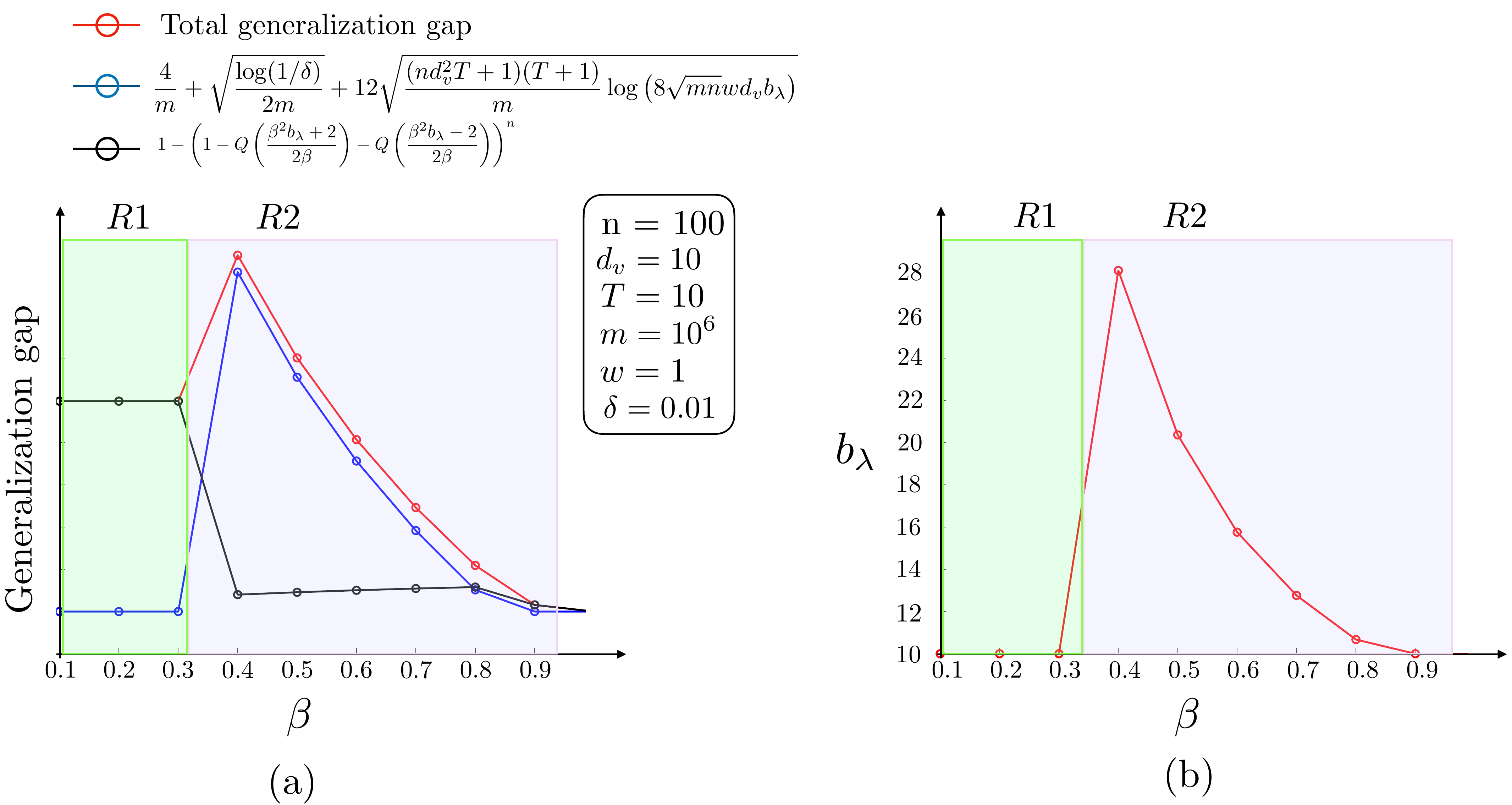}
	\caption{(a) The total generalization gap from Theorem \ref{Proposition-2: unbounded-input-channel-snr}, generalization gap from Theorem \ref{Theorem-1: dudley-entropy-intergral}, and the generalization gap due to unbounded log-likelihood ratio as a function of the channel SNR, (b) Selecting the bound on LLR ($b_{\lambda}$) to minimize the generalization gap. \label{fig: Theorem-3-remark}}
		\end{center}
  \vspace{-20pt}
\end{figure}

\section{Experimental Results}
\label{sec:experimental-results}
\noindent In this section, we present some numerical results to complement our theoretical bounds.  We consider binary phase shift keying (BPSK) modulation and AWGN channel, and the received channel output for $1 \leq i \leq n$ is given as  $\mathbf{y}[i] = (-1)^{\mathbf{x}[i]} + \mathbf{z}[i]$. We focus on Tanner codes with: (i) $n = 155$, $k = 64$, $d_v = 3$, $d_c = 5$; (ii) $n = 310$, $k = 128$, $d_v = 3$, $d_c = 5$ and study the empirical generalization performance of NBP decoders whose architecture was proposed in \cite{nachmani2016learning}, and also described in Section \ref{sec: preliminaries} of this paper. 
We adopt the software provided with the papers \cite{lugosch2017neural, nachmani2018deep} for our experiments.
We train the weights of the  NBP decoder until convergence by minimizing the cross-entropy loss between the true and the predicted codeword. 
We use ADAM optimizer for training with a learning rate of $0.01$.
We evaluate the NBP decoder by measuring the generalization gap (difference between  average BER attained on the test and training datasets). 
We perform each experiment for $10$ trials, and the distribution of the generalization gap over these $10$ randomized runs are plotted on a boxplot.
We next discuss the impact of the dataset size ($m$), and the decoding iterations ($T$) on the generalization gap. 

\noindent \textbf{a. Impact of training dataset size ($m$):} We consider the NBP decoder with $T=3$ decoding iterations (equivalently, $6$ layers) trained for channel SNR of $2$ dB; we vary the training data set size from $m=10^3$ to $m=10^4$ in steps of $1000$.  From the results in Fig. \ref{fig: generr_m_exp}(a), (b), we observe that the generalization gap is the largest for $m = 1000$, and generally decays with $m$.
For a smaller dataset size, the overfitting on the training samples is severe. 
Therefore, the NBP decoder fails to generalize on unseen samples in the test data. We also repeated the above experiment for various values of $T$ as well as by changing SNR. We found the inverse monotonic dependence on $m$ to be consistent across different values of $T$ and SNR (plots are  omitted due to lack of space).


\begin{figure}[!t]
	\begin{center}
		\includegraphics[scale = 0.25]{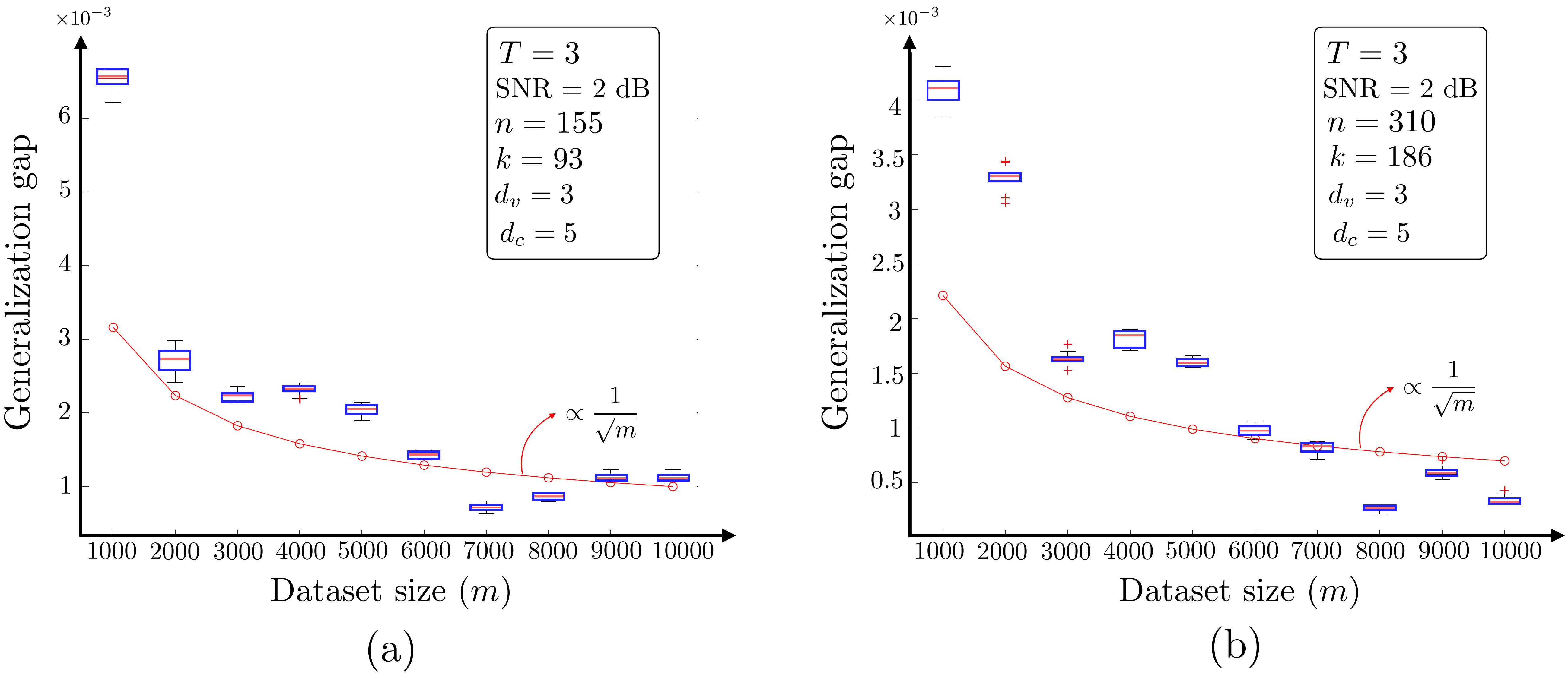}
	\caption{Generalization gap as a function of the dataset size $m$ at channel SNR = $2$ dB for (a) Tanner code with $n = 155$, and $k = 93$, (b) Tanner code with $n = 310$, and $k = 186$. \label{fig: generr_m_exp}}
		\end{center}
\end{figure}

 \begin{figure}[!t]
	\begin{center}
		\includegraphics[scale = 0.25]{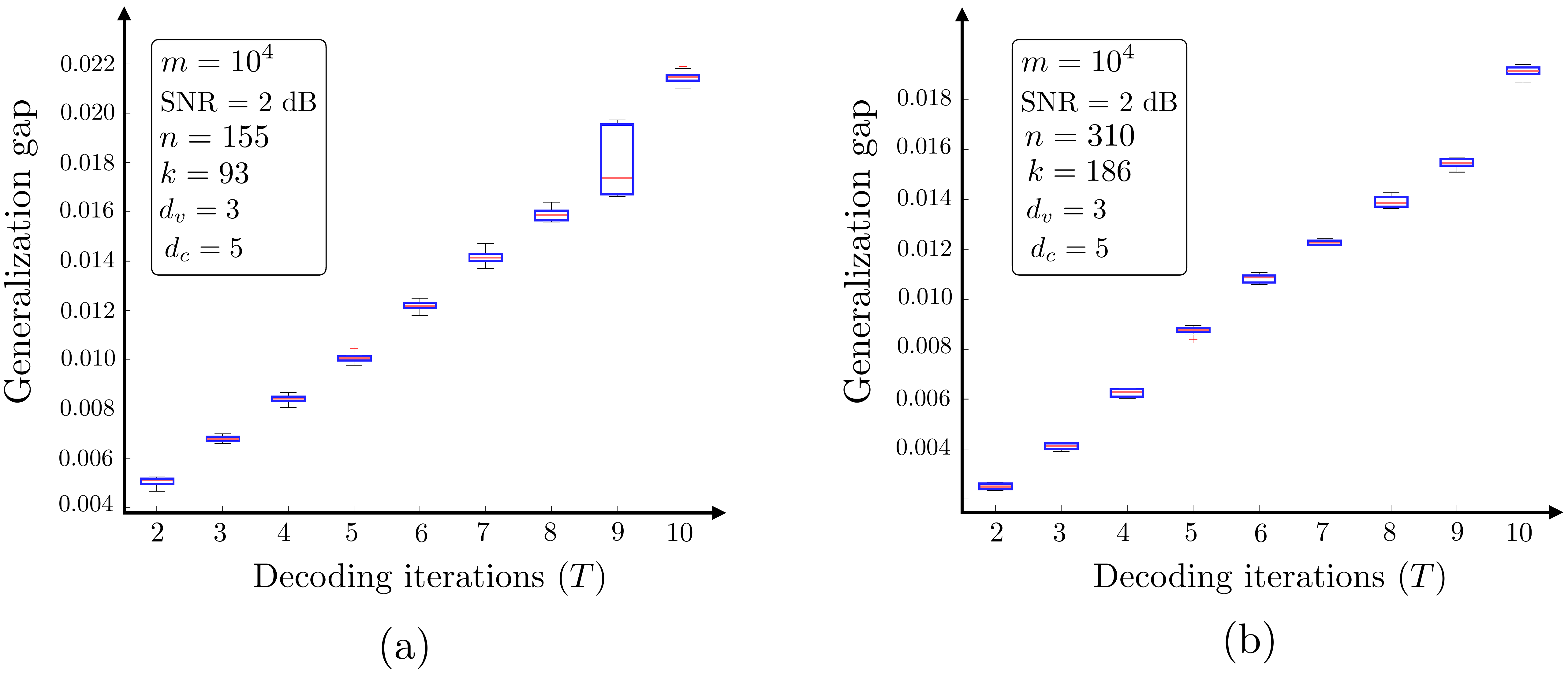}
\caption{Generalization gap as a function of the decoding iterations $T$ ($\propto$ number of layers) at channel SNR = $2$ dB for (a) Tanner code with $n = 155$, and $k = 93$, (b) Tanner code with $n = 310$, and $k = 186$. \label{fig: gen_err_T}}
 \vspace{-0.8cm}
		\end{center}
\end{figure}

\noindent \textbf{b. Impact of decoding iterations ($T$):}
In this experiment, we study the impact of decoding iterations (which is proportional to the number of hidden layers) in the NBP decoder on the generalization gap.
Here, we fixed channel SNR of $2$ dB, training dataset size $m = 10^4$ and varied $T$ from $\{2,3,\ldots, 10\}$. 
As seen in Fig. \ref{fig: gen_err_T} the generalization gap grows linearly with $T$, which is consistent with Theorem \ref{Theorem-1: dudley-entropy-intergral} (which behaves as $\mathcal{O}(T)$).  
 Increasing the number of parameters will cause overfitting of the NBP decoder resulting in a larger generalization gap.
We note that this observation (i.e., linear dependence on $T$) was consistent for different dataset sizes, and channel SNR values. 

 \begin{figure}[!t]
	\begin{center}
		\includegraphics[scale = 0.25]{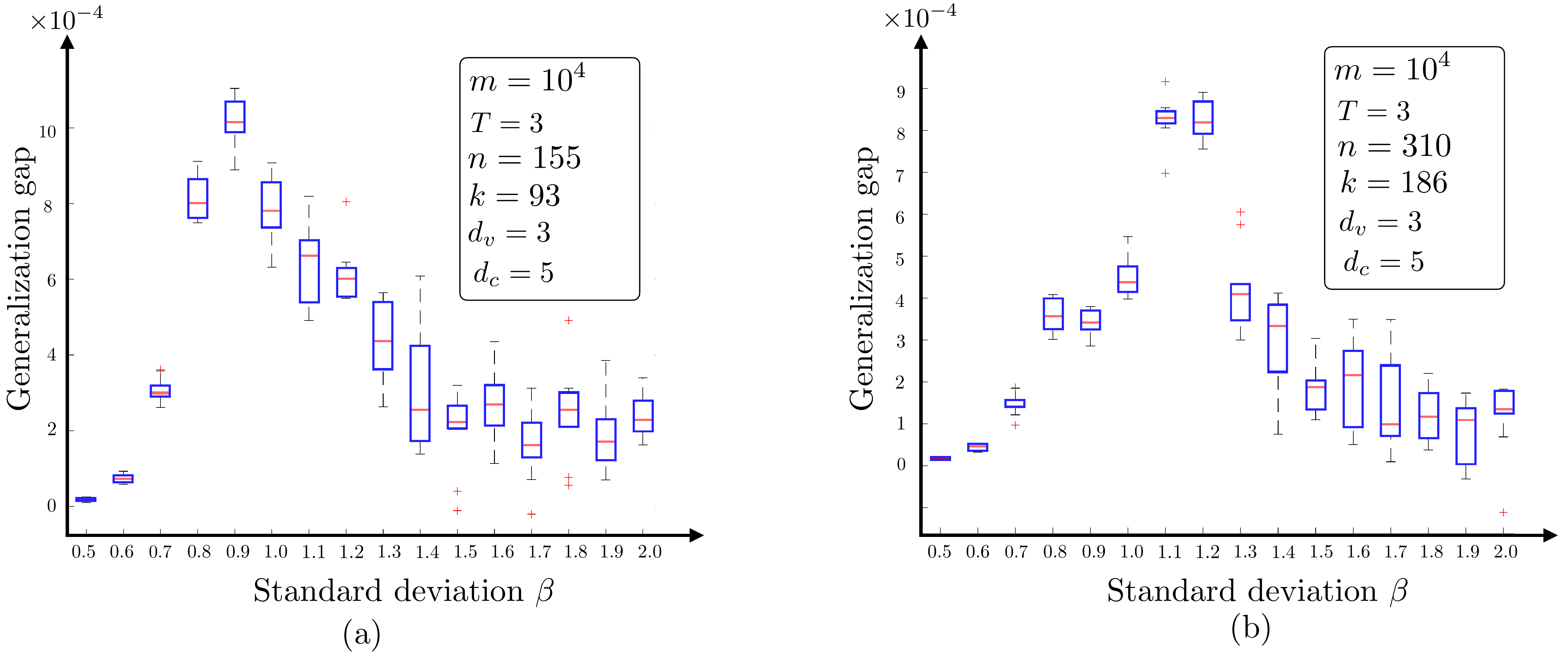}
\caption{Generalization gap as a function of the standard deviation of the Gaussian noise $\beta$ for (a) Tanner code with $n = 155$, and $k = 93$, (b) Tanner code with $n = 310$, and $k = 186$. \label{fig: gen_err_snr}}
 \vspace{-0.8cm}
		\end{center}
\end{figure}

\noindent \textbf{c. Impact of standard deviation of the Gaussian noise ($\beta$):}
In this experiment, we examine the effects of varying the standard deviation of Gaussian noise, represented by $\beta$, on the generalization gap of NBP decoders. 
We consider NBP decoders with $T=3$ iterations and trained on a dataset of size $m=10^4$.
According to the results depicted in Fig. \ref{fig: gen_err_snr}(a) and (b), we observe that the generalization gap exhibits non-monotonic behavior in relation to
$\beta$. Specifically, for $\beta \leq 1$, there is an increase in the generalization gap. Conversely, for $\beta \geq 1$, the generalization gap begins to decrease.
When $\beta$ is substantially large, the channel output becomes statistically independent of the input due to the increased channel noise. Consequently, this results in higher training and test BER, which in turn leads to a reduced generalization gap.

\noindent \textbf{d. Impact of blocklength ($n$):}
To study the impact of blocklength keeping the variable node degree fixed, we use Progressive Edge-Growth (PEG) algorithm \cite{hu2001progressive, hu2005regular} to construct two quasi-cyclic  LDPC (QC-LDPC) parity check matrices with: (i) $n = 680$, $k = 340$, $d_v = 3$, $d_c = 6$; (ii) $n = 1000$, $k = 500$, $d_v = 3$, $d_c = 7$.
We use the two codes as the parent code, and derive the parity check matrices for varying blocklengths keeping the variable node degree fixed by masking the columns of the parity-check matrix of the parent codes.
Specifically, in Fig. \ref{fig: gen_err_n} (a) we consider codes with blocklengths  $n = 425$, $510$, $595$ derived from QC-LDPC parity check matrix with $n = 680$, $k = 340$; and in Fig. \ref{fig: gen_err_n} (b) we consider codes with blocklengths  $n = 600$, $700$, $800$, $900$ derived from QC-LDPC parity check matrix with $n = 1000$, $k = 500$.
We note that in addition to the blocklength, the descendent codes have different code-rates than the parent QC-LDPC code.
However, this method allows us to generate descendent codes that have the same structure (or same nature) as that of the parent code, thereby, eliminating the impact of varying code structures on the generalization gap. 
To account for the different training bit-errors for codes with varying code rates, we normalize the generalization gap with the training bit-error-rate (or the empirical risk). 
We also plot the theoretical bound derived in Theorem \ref{Theorem-1: dudley-entropy-intergral} normalized with the training bit-error-rate  (i.e., ${12} \sqrt{\frac{(n d^2_v T + 1)({T+1})}{m} \log \left(8{\sqrt{mn}wd_vb_{\lambda}}\right)}/ {\hat{R}_{\text{BER}}(f)}$).
We consider NBP decoder with $T= 3$ decoding iterations trained for channel SNR of $2$ dB using dataset whose size $m = 10^4$. 
As seen in Fig. \ref{fig: gen_err_n} the generalization gap grows with $n$, which is consistent with Theorem \ref{Theorem-1: dudley-entropy-intergral}. 
The implication of this result is that for a fixed set of parity check equations, the decoding complexity increases with the blocklength; and the generalization gap is expected to increase with the blocklength for the codes with same structure.  


 \begin{figure}[!t]
	\begin{center}
		\includegraphics[scale = 0.25]{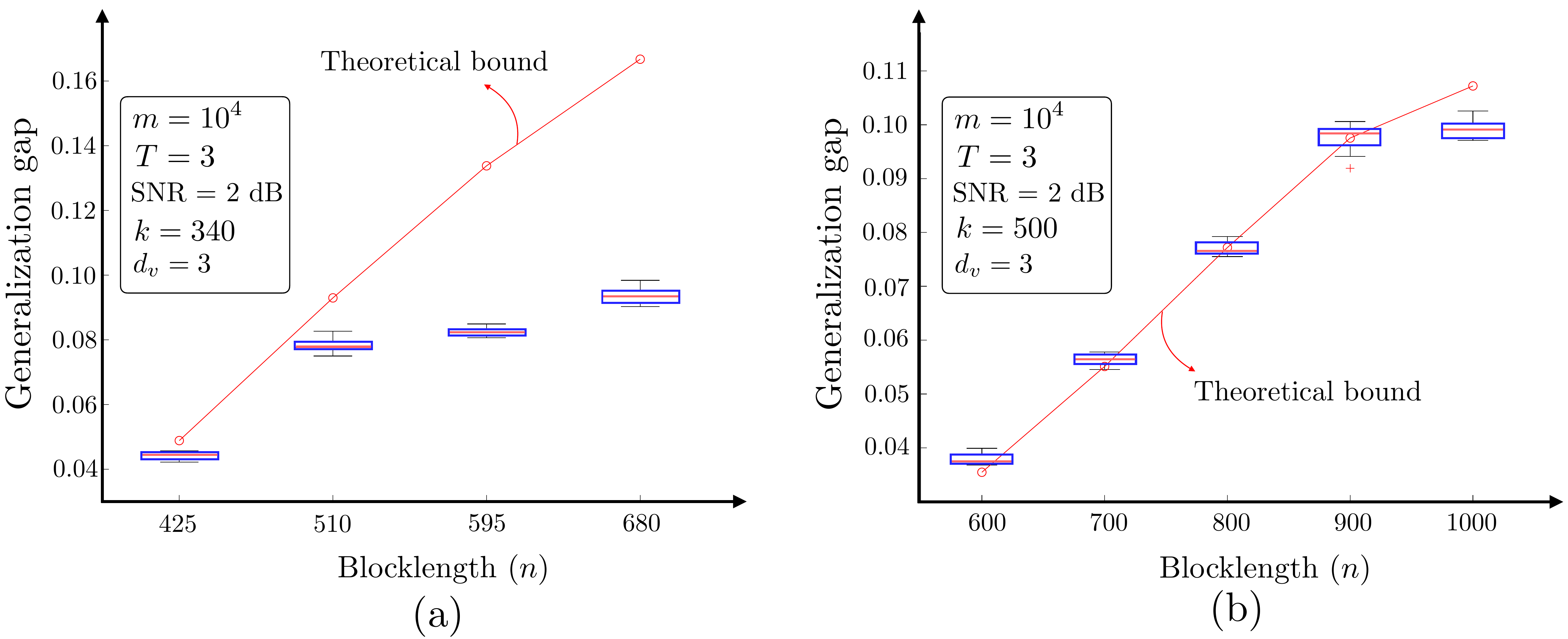}
\caption{Generalization gap as a function of the blocklength $n$ at channel SNR = $2$ dB for parent QC-LDPC codes with (a) $n = 1000$, and $k = 500$, (b) $n = 680$, and $k = 340$. The descendent codes with shorter blocklengths are derived by masking columns of parity check matrix in parent code.  \label{fig: gen_err_n}}
 \vspace{-0.8cm}
		\end{center}
\end{figure}

\section{Conclusions}
\label{sec:conclusion}

\noindent In this work, we presented results on the generalization gap of NBP decoders as a function of training dataset size, decoding iterations and code parameters (such as blocklength, message length, variable node degree, and parity check node degree).
We utilize the PAC-learning theory to express the generalization gap as a function of the Rademacher complexity of class of NBP decoders.
The sparse connections in the NBP decoder architecture plays a critical role in further upper bounding the Rademacher complexity term as a function of the code parameters.
Our bounds exhibit mild polynomial dependence on the blocklength $n$ and the decoding iterations (layers) $T$. 
To the best of our knowledge, our work is the first to provide theoretical guarantees for NBP decoders corresponding to both regular and irregular parity check matrices.
We also present generalizations of our theoretical result to account for different channel characteristics. 
We also presented comprehensive set of experiments using Tanner codes and Quasi-cyclic LDPC codes to evaluate our theoretical bounds in this paper.
In our empirical results, we observe that the generalization gap increases linearly with the decoding iterations, and grows with the blocklength. 
In addition, we observe that the generalization gap decays with the training dataset size, thereby supporting the theoretical results in this paper.
There are several interesting directions for future work, including a) obtaining generalization gap bounds for NBP decoders when the decoder is trained and tested over a range of SNR values; b) obtaining generalization bounds for ML based decoders with practical constraints (such as quantized weights); c) extending the ideas for other type of ML based codes/decoders (i.e., beyond BP type decoder architectures);  d) a framework to select the code parameters (i.e., blocklength and variable node degree pairs) that minimize the generalization gap for a given channel SNR is also an important research direction.

\bibliographystyle{ieeetr}
\bibliography{refs.bib}

\appendices
\section{Proof of Proposition \ref{Proposition-1: bit-wise-Rademacher}}
\label{sec: Proposition-1-Proof}
\noindent The proof of Proposition \ref{Proposition-1: bit-wise-Rademacher} directly follows from the Definition \ref{def-1:empirical-Rademacher-complexity} of empirical Rademacher complexity.
\begin{align}
    & R_{m} (\mathcal{F}_{L,T}) = \underset{\sigma}{\mathbb{E}} \left[\underset{{f \in \mathcal{F}_{T}}}{\sup} \frac{1}{m}\sum_{i = 1}^{m}\sigma_{i} l_\text{BER}(f({\bm{\lambda}_i}),\mathbf{x}_i)\right] \nonumber\\
    & {=}\underset{\sigma}{\mathbb{E}}\left[\sup_{{f \in \mathcal{F}_{T}}}\frac{1}{m}\sum_{i=1}^m\sigma_i \frac{d_H(f({\bm{\lambda}}_i),\mathbf{x}_i)}{n}\right]\nonumber\\
    &\stackrel{(a)}{=}\underset{\sigma}{\mathbb{E}}\left[\sup_{{f \in \mathcal{F}_{T}}}\frac{1}{mn}\sum_{i=1}^m\sigma_i \sum_{j=1}^n\left(\frac{1-f({\bm{\lambda}}_i){[j]}}{2}-\frac{1-\mathbf{x}_{i}{[j]}}{2}\right)^2\right]\nonumber\\
    &\stackrel{(b)}{=}\frac{1}{2}\underset{\sigma}{\mathbb{E}}\left[\sup_{{f \in \mathcal{F}_{T}}}\frac{1}{mn}\sum_{i=1}^m\sigma_i \sum_{j=1}^n\left(1-f({\bm{\lambda}}_i){[j]} \cdot {\mathbf{x}_{i}}{[j]}\right)\right]\nonumber\\
    & {=}\frac{1}{2}\underset{\sigma}{\mathbb{E}}\left[\sup_{{f \in \mathcal{F}_{T}}}\frac{1}{mn} \sum_{i=1}^m\sum_{j=1}^n\left(-f({\bm{\lambda}}_i){[j]} \cdot (\sigma_{i} \cdot \mathbf{x}_{i}{[j]})\right)\right].
    \label{eq: proposition-1-eq1}
\end{align}
where, in step (a), we consider the mapping such that $\mathbf{x}_i[j] \in \{-1, 1\}$, and $f({\bm{\lambda}}_i[j]) \in \{-1, 1\}$; and in step (b), we express the Hamming distance $d_H(f(\bm{\lambda}_i),\mathbf{x}_i) = \frac{1}{n}\sum_{j=1}^n\left(1-f({\bm{\lambda}}_i){[j]} \times  {\mathbf{x}_{i}}{[j]}\right)$. 
In what follows, we take the supremum over $\mathcal{F}_T$ inside the summation. 
We note that the distribution of $-\sigma_i \mathbf{x}_{i}{[j]}$ and  $\sigma_i$ is the same, and we obtain: 
\begin{align}
    R_{m} (\mathcal{F}_{L,T}) 
    &{\leq} \frac{1}{2}\underset{\sigma}{\mathbb{E}}\left[\sum_{j=1}^n\sup_{{f \in \mathcal{F}_{T}}}\frac{1}{mn} \sum_{i=1}^m\left(-f({\bm{\lambda}}_i){[j]} \cdot (\sigma_{i} \cdot \mathbf{x}_{i}{[j]})\right)\right]\nonumber\\
    &=\frac{1}{2n}\sum_{j=1}^n \underset{\sigma}{\mathbb{E}}\left[\sup_{{f \in \mathcal{F}_{T}}}\frac{1}{m}\sum_{i=1}^m\sigma_i \cdot f({\bm{\lambda}}_i){[j]}\right]\nonumber\\
    &=\frac{1}{2n} \sum_{j = 1}^{n} R_m(\mathcal{F}_{T}[j]).
    \label{eq: proposition-1-eq2}
\end{align}
In  \eqref{eq: proposition-1-eq2}, the last step follows from the Definition \ref{def-2:bit-wise-Rademacher-complexity}. 

\section{Proof of Theorem \ref{Theorem-1: dudley-entropy-intergral}}
\label{sec: Theorem-1-Proof}
\noindent In this appendix, we present the proof of Theorem \ref{Theorem-1: dudley-entropy-intergral}.
 We first define the covering number and packing number of $\mathcal{F}_{T}$ which is the set of all NBP decoders with $T$ decoding iterations.
\begin{definition} \textbf{(Covering Number)}
The covering number $\mathcal{N}(\mathcal{F}_T ,\epsilon , \|\cdot\|_k)$ of the set $\mathcal{F}_{T}$ with respect to the $k$-{th} norm for $\epsilon > 0$ is defined as
\begin{align}
     \mathcal{N}(\mathcal{F}_T ,\epsilon , \|\cdot\|_k) = \underset{n}{\min} |\{g_1,\cdots,g_n\}| \label{definition-1: covering-number-a}, & \\  \text{s.t. } \underset{1 \leq i \leq n}{\min} \|{f}(\bm{\lambda})-{g_i}(\bm{\lambda})\|_k \leq \epsilon \label{definition-1: covering-number-b}. &
\end{align}
\eqref{definition-1: covering-number-b} must be satisfied for any $f \in \mathcal{F}_T$ and input log likelihood ratio $\bm{\lambda}$; then the set $\{{g_1},\cdots,{g_n}\} \subseteq \mathcal{F}_T$ is the $\epsilon$-cover of $\mathcal{F}_T$. 
\label{definition-1: covering-number}
\end{definition}
\begin{definition} \textbf{(Packing number)} For a set of NBP decoders $\mathcal{F}_T$, its packing number $\mathcal{M}(\mathcal{F}_T ,\epsilon , \|\cdot\|_k)$ with respect to the $k$-{th} norm for $\epsilon > 0$ is defined as
\begin{align}
    \mathcal{M}(\mathcal{F}_T ,\epsilon , \|\cdot\|_k) = \underset{n}{\max} |\{g_1,\cdots,g_n\}| \label{definition-1: packing-number-a}, & \\  \text{s.t. }   \|{g_i}(\bm{\lambda})-{g_j}(\bm{\lambda})\|_k > \epsilon \label{definition-1: packing-number-b}. &
\end{align}
\eqref{definition-1: packing-number-b} must be satisfied for any $g_i,g_j \in \{g_1,\cdots,g_n\}$; and the set $\mathcal{P}(\mathcal{F}_T ,\epsilon , \|\cdot\|_k) = \{g_1,\cdots,g_n\} \subset \mathcal{F}_T$ that satisfies \eqref{definition-1: packing-number-a}, \eqref{definition-1: packing-number-b} is called the $\epsilon$-packing of $\mathcal{F}_T$. 
\label{definition-2: packing-number}
\end{definition}

\noindent From Proposition \ref{Proposition-1: bit-wise-Rademacher}, we have that, 
    $\mathcal{R}_\text{BER}(f)
    \leq \mathcal{\hat{R}}_{\text{BER}}(f) + \frac{1}{2n} \sum_{j = 1}^{n} R_m(\mathcal{F}_{T}[j]) +\sqrt{\frac{\log(1/\delta)}{2m}}$.
Next, we use a PAC-Learning approach to bound the bit-wise Rademacher complexity term $R_m(\mathcal{F}_{T}[j])$ as a function of $m$, and spectral norm of the weight matrices of the NBP decoder.
We use similar reasoning to that used in generalization bound results for graph neural networks and recurrent neural networks in \cite{garg2020generalization, chen2019generalization}; and we can adopt Lemma A.5. in \cite{bartlett2017spectrally} to bound the bit-wise Rademacher complexity as:
\begin{align}
    R_m(\mathcal{F}_{T}[j]) \leq \underset{\alpha > 0}{\inf} \left( \frac{4\alpha}{\sqrt{m}} + \frac{12}{m} \int\limits_{\alpha}^{\sqrt{m}} \sqrt{\log \mathcal{N}(\mathcal{F}_{T}[j],\epsilon, \|\cdot\|_2)} d\epsilon \right).
    \label{eq: Rademacher upper bound bartlett}
\end{align}

\noindent To further upper bound the bit-wise Rademacher complexity, we make use of the upper bounds on the spectral norm of the weight matrices.
 Recall our assumption made in Section \ref{sec: main-results} that the maximum absolute value of the non-zero entries in the weight matrices are bounded by a non-negative constant $w$.
In addition, we also use the result in \cite{mathias1990spectral} that the spectral norm of any matrix $\mathbf{A}$ can be upper bounded as a function of its maximum absolute column sum norm $\|\mathbf{A}\|_1$, and maximum absolute row sum norm $\|\mathbf{A}\|_{\infty}$  as 
$\|\mathbf{A}\|_2 \leq \sqrt{\|\mathbf{A}\|_{\infty}\|\mathbf{A}\|_{1}}$.
Using the assumption and the result in \cite{mathias1990spectral}, the spectral norm of the weight matrices $\mathbf{W}^{(t)}_1$, $\mathbf{W}^{(t)}_2$ for any $ 1 \leq t \leq T$ can be upper bounded as: 
\begin{align}
   & B_{W_1} = \|\mathbf{W}^{(t)}_1\|_2 \leq w \sqrt{d_v}, \nonumber \\ 
   & B_{W_2} = \|\mathbf{W}^{(t)}_2\|_2 \leq w {(d_v-1)}. 
   \label{eq: spectral norm bounds w1-w2}
\end{align}
Similarly, we obtain the $L_2$ norm bounds on for any $j-$th row vector in matrices $\mathbf{W}_3$ and $\mathbf{W}_4$ as:
\begin{align}
   & B_{w_3} = \|\mathbf{W}_3[j,:]\|_2 \leq w \sqrt{d_v},  \nonumber \\ 
   & B_{w_4} = \|\mathbf{W}_4[j,:]\|_2 \leq w.
   \label{eq: spectral norm bounds w3-w4}
\end{align}

 \noindent We use these spectral norm bounds in Lemma \ref{lemma-1: f is lipschitz} in which we show that the NBP decoder is Lipschitz continuous with respect to its weight matrices. That is, for the $j$-{th} bit we have, 
\begin{align}
\left\|f({\bm{\lambda}}){[j]} - f^{\prime}({\bm{\lambda}}){[j]}\right\|_{2} 
 \leq  & \sum^{T}_{i=1}  \rho_{_{W^{(i)}_1}} \left\|\mathbf{W}^{(i)}_1 -\mathbf{W}_1^{\prime {(i)}}\right\|_{2}  + 
  \sum^{T}_{i=2}  \rho_{_{W^{(i)}_2}} \left\|\mathbf{W}^{(i)}_2 -\mathbf{W}_2^{\prime {(i)}}\right\|_{2}
  \nonumber \\
 & +  \rho_{_{w_3}} \left\|\mathbf{W}_3[j,:]-\mathbf{W}^{\prime}_3[j,:]\right\|_{2}   + \rho_{_{w_4}} \left\|\mathbf{W}_4[j,:]- \mathbf{W}^{\prime}_4[j,:]  \right\|_{2}.
\label{eq: lemma-4-lipschitz-result}
\end{align} 
In \eqref{eq: lemma-4-lipschitz-result}, $\rho_{_{W_1}}$, $\rho_{_{W_2}}$, $\rho_{_{w_3}}$, $\rho_{_{w_4}}$ are Lipschitz parameters that is a function of the spectral norm bounds of the weight matrices, the $L_2$ norm bound of input $\bm{\lambda}$, and are given as,     
\begin{align}
&\rho_{_{W^{(i)}_1}}  =   n  b_{\lambda} B_{w_3} \left(\sqrt{n} B_{W_2} \right)^{T-i}, 
\nonumber \\
&\rho_{_{W^{(i)}_2}}  =  n T b_{\lambda} B_{W_1} B_{w_3}  \frac{\left(\sqrt{n}  B_{W_2} \right)^{T-1}-1}{\sqrt{n}  B_{W_2}-1}+ n b_{\lambda} B_{W_1} B_{w_3}\left(B_{W_2}\right)^{T-2} \frac{(\sqrt{n})^{T-1}-1}{\sqrt{n}-1},
\nonumber \\
&\rho_{_{w_3}}  = \sqrt{n}B_{W_1} b_{\lambda} \left( \frac{\left( B_{W_2}\right)^{T-1}-1}{ B_{W_2} - 1} +   \left(B_{W_2}\right)^{T-1} \right), 
\nonumber \\
&\rho_{_{w_4}}  = b_{\lambda}.
\label{eq: w1w2w3w4-lipschitz}
\end{align}
As a result of \eqref{eq: lemma-4-lipschitz-result}, the covering number $\mathcal{N}(\mathcal{F}_{T}[j], {\epsilon}, \|\cdot\|_2)$ in \eqref{eq: Rademacher upper bound bartlett} can be upper bounded using the covering number of all distinct weight matrices as follows,
\begin{align}
    \mathcal{N}(\mathcal{F}_T[j], {\epsilon}, \|\cdot\|_2) \leq & 
    \prod_{i = 1}^{T} \mathcal{N}(\mathbf{W}^{(i)}_1, \frac{\epsilon}{(2T + 1)\rho_{_{W^{(i)}_1}}}, \|\cdot\|_F) \times \prod_{i = 2}^{T}  \mathcal{N}(\mathbf{W}^{(i)}_2, \frac{\epsilon}{(2T + 1)\rho_{_{W^{(i)}_2}}}, \|\cdot\|_F)
     \nonumber \\
     &  \times \mathcal{N}(\mathbf{W}_3[j,:], \frac{\epsilon}{(2T + 1)\rho_{_{w_3}}}, \|\cdot\|_F) \times \mathcal{N}(\mathbf{W}_4[j,:], \frac{\epsilon}{(2T + 1)\rho_{_{w_4}}}, \|\cdot\|_F)  
     \label{eq:Ft-upper-bound-matrices}
\end{align}

  \noindent 
 We now upper bound the covering number of the set of decoders $\mathcal{F}_{T}$ in \eqref{eq:Ft-upper-bound-matrices} using Lemma \ref{Lemma-2: Covering number matrices}, in which we derive an upper bound on the covering number of sparse matrices as a function of the number of rows, columns, number of non-zero entries in the sparse matrix, and the spectral norm bound.

\noindent \textbf{Bounding $\mathcal{N}(\mathcal{F}_T[j], {\epsilon}, \|\cdot\|_2)$:}
We know that the matrices $\mathbf{W}_1$,  $\mathbf{W}_2$, $\mathbf{W}_3$, and $\mathbf{W}_4$ have $1$, $d_v - 1$, $d_v$, and $1$ non-zero entries in each row, respectively.  
Using Lemma \ref{Lemma-2: Covering number matrices}, the covering number of the weight matrices of the NBP decoder $f$ for $1 \leq  i \leq T$ can be bounded as follows:
\begin{align}
    & \mathcal{N}(\mathbf{W}^{(i)}_1, \frac{\epsilon}{(2T + 1) \rho_{_{W^{(i)}_1}}}, \|\cdot\|_F)  \leq   {\left(1 + \frac{ (4T + 2) \sqrt{n} B_{W_1}\rho_{_{W^{(i)}_1}}}{\epsilon}\right)^{n d_v}}  \nonumber \\
    & \mathcal{N}(\mathbf{W}^{(i)}_2, \frac{\epsilon}{(2T + 1) \rho_{_{W^{(i)}_2}}}, \|\cdot\|_F) \leq {\left(1 + \frac{ (4T + 2)  \sqrt{nd_v} B_{W_2}\rho_{_{W^{(i)}_2}}}{\epsilon}\right)^{(d_v - 1) nd_v}}  \nonumber \\
    & \mathcal{N}(\mathbf{W}_3[j,:], \frac{\epsilon}{(2T + 1) \rho_{_{w_3}}}, \|\cdot\|_2) \leq {\left(1 + \frac{ (4T + 2)  B_{w_3}\rho_{_{w_3}}}{\epsilon}\right)^{d_v}}   \nonumber \\
    & \mathcal{N}(\mathbf{W}_4[j,:], \frac{\epsilon}{(2T + 1) \rho_{_{w_4}}}, \|\cdot\|_2) \leq {\left(1 + \frac{ (4T + 2)  B_{w_4}\rho_{_{w_4}}}{\epsilon}\right)}
    \label{eq: covering-number-matrices}
\end{align}

\noindent Substituting \eqref{eq: covering-number-matrices} in \eqref{eq:Ft-upper-bound-matrices}, we obtain,
\begin{align}
    \mathcal{N}(\mathcal{F}_T[j], {\epsilon}, \|\cdot\|_2)
      \leq &
      {\prod_{i = 1}^{T} \hspace{-2pt} \left(\hspace{-2pt}1 \hspace{-2pt}+\hspace{-2pt} \frac{ (4T + 2) \sqrt{n} B_{W_1}\rho_{_{W^{(i)}_1}}}{\epsilon}\hspace{-2pt}\right)^{n d_v}} \hspace{-10pt}  \times  {\prod_{i = 2}^{T} \hspace{-2pt} \left(\hspace{-2pt} 1 + \frac{ (4T + 2) \sqrt{nd_v} B_{W_2}\rho_{_{W^{(i)}_2}}}{\epsilon}\hspace{-2pt}\right)^{(d_v - 1) nd_v}} 
     \nonumber \\&  \times {\left(1 + \frac{ (4T + 2)  B_{w_3}\rho_{_{w_3}}}{\epsilon}\right)^{d_v}}  \times  {\left(1 + \frac{ (4T + 2) B_{w_4}\rho_{_{w_4}}}{\epsilon}\right)}. 
     \label{eq: covering-number-matrices-prod}
\end{align}

\noindent Substituting the values of Lipschitz parameters obtained in \eqref{eq: w1w2w3w4-lipschitz}, and assuming $nd_v > (n+1)$ , the term $\mathcal{N}(\mathcal{F}_T[j], {\epsilon}, \|\cdot\|_2)$ can be bounded as: 
\begin{align}
\mathcal{N}(\mathcal{F}_T[j], {\epsilon}, \|\cdot\|_2) \leq  \left(1 + \frac{(4T + 2) \sqrt{nd_v}\left(c_1 + c_2 \right)}{\epsilon}  \right)^{(n d^2_v T + 1)}. 
\label{eq: covering-number-Ft-using-matrices}
\end{align}
where, $c_1 = nT B_{W_1} B_{W_2} B_{w_3} b_{\lambda} \frac{( \sqrt{n} B_{W_2})^{T-1}-1}{\sqrt{n} B_{W_2}-1}$,
 $c_2 =  n B_{W_1} B_{w_3} b_{\lambda} \left(\sqrt{n} B_{W_2} \right)^{T-1}$.
 This value of $c_1$, and $c_2$ are obtained for large enough $n$, $T$, and $d_v$, which is an assumption that we make. 

\noindent We now make use of the spectral norm bounds in \eqref{eq: spectral norm bounds w1-w2} and \eqref{eq: spectral norm bounds w3-w4} in \eqref{eq: covering-number-Ft-using-matrices}; and further upper bound the covering number $\mathcal{N}(\mathcal{F}_T[j], {\epsilon}, \|\cdot\|_2)$ as,
\begin{align}
\mathcal{N}(\mathcal{F}_T[j], {\epsilon}, \|\cdot\|_2) \leq \left(1 + \frac{ (4T+2)\sqrt{nd_v}}{\epsilon} b_{\lambda} (T+1) \left(\sqrt{n}wd_v\right)^{T+1} \right)^{(n d^2_v T + 1)}
\label{eq: covering-number-Ft-using-matrices-2}
\end{align}

\noindent In \eqref{eq: Rademacher upper bound bartlett}, the integral can be further upper bounded using \eqref{eq: covering-number-Ft-using-matrices-2} and we have that, 
\begin{align}
     \int\limits_{\alpha}^{\sqrt{m}} \sqrt{\log \mathcal{N}(\mathcal{F}_{T}[j],\epsilon, \|\cdot\|_2)} d\epsilon & \leq   \int\limits_{\alpha}^{\sqrt{m}} \sqrt{{(n d^2_v T + 1)} \text{term}_1} d\epsilon \leq \sqrt{{m(n d^2_v T + 1)} \text{term}_2}
   \label{eq: covering intergral upper bound}
 \end{align}
 where, $\text{term}_1 \hspace{-2pt} = \hspace{-2pt} \log  \hspace{-2pt} \left( \hspace{-2pt} \frac{8 \sqrt{nd_v}}{\epsilon} b_{\lambda} (T+1)^2 \left(\sqrt{n}wd_v\right)^{T+1} \hspace{-1pt}\right) \hspace{-1pt}$; $\text{term}_2 \hspace{-2pt} = \hspace{-2pt} \log  \hspace{-2pt} \left( \hspace{-2pt} {8 \sqrt{mnd_v}} b_{\lambda} (T+1)^2 \left(\sqrt{n}wd_v\right)^{T+1} \hspace{-1pt} \right)\hspace{-1pt}$.
Substituting \eqref{eq: covering intergral upper bound} in \eqref{eq: Rademacher upper bound bartlett} 
and assuming that the term $(\sqrt{m}b_{\lambda})^{T+1}$ is large enough to approximate the term inside the logarithm, we have that, 
\begin{align}
     \mathcal{R}_\text{BER}(f) &
    - \mathcal{\hat{R}}_{\text{BER}}(f) 
    \leq  \frac{4}{{m}} +\sqrt{\frac{\log(1/\delta)}{2m}} + {12} \sqrt{\frac{(n d^2_v T + 1)({T+1})}{m} \log \left(8{\sqrt{mn}wd_vb_{\lambda}}\right)} .
\end{align}

\section{Lipschitzness in NBP decoders}
\label{sec: Lemma-1-Proof}
 \begin{lemma}
 \label{lemma-1: f is lipschitz}
 For $n$ length codeword, the bit-wise output of the NBP decoder $f \in \mathcal{F}_T$ is Lipschitz in its weight matrices $\mathbf{W}_1,\mathbf{W}_2,\mathbf{W}_3,\mathbf{W}_4$ such that,
\begin{align*}
 \left\|f({\bm{\lambda}}){[j]} -f^{\prime}({\bm{\lambda}}){[j]}\right\|_{2} 
& \leq  \sum^{T}_{i=1}  \rho_{_{W^{(i)}_1}} \left\|\mathbf{W}^{(i)}_1 -\mathbf{W}_1^{\prime {(i)}}\right\|_{2} + 
  \sum^{T}_{i=2}  \rho_{_{W^{(i)}_2}} \left\|\mathbf{W}^{(i)}_2 -\mathbf{W}_2^{\prime {(i)}}\right\|_{2} \nonumber \\ 
 & +  \rho_{_{w_3}} \left\|\mathbf{W}_3[j,:]-\mathbf{W}^{\prime}_3[j,:]\right\|_{2}  +\rho_{_{w_4}} \left\|\mathbf{W}_4[j,:]- \mathbf{W}^{\prime}_4[j,:]  \right\|_{2}.
\end{align*}
The coefficients $\rho_{_{W_1}}$, $\rho_{_{W_2}}$, $\rho_{_{w_3}}$ and $\rho_{_{w_4}}$ are as follows:  
\begin{align}
&\rho_{_{W^{(i)}_1}}  =   {n  b_{\lambda} B_{w_3}} \left(\sqrt{n} B_{W_2} \right)^{T-i}, 
\nonumber \\
&\rho_{_{W^{(i)}_2}}  =  n T b_{\lambda} B_{W_1} B_{w_3}  \frac{\left(\sqrt{n}  B_{W_2} \right)^{T-1}-1}{\sqrt{n}  B_{W_2}-1} + n b_{\lambda} B_{W_1} B_{w_3}\left(B_{W_2}\right)^{T-2} \frac{(\sqrt{n})^{T-1}-1}{\sqrt{n}-1},
\nonumber \\
&\rho_{_{w_3}}  = \sqrt{n}B_{W_1} b_{\lambda} \left( \frac{\left( B_{W_2}\right)^{T-1}-1}{ B_{W_2} - 1} +   \left(B_{W_2}\right)^{T-1} \right), 
\nonumber \\
&\rho_{_{w_4}}  = b_{\lambda}.
\end{align}
\label{Lemma-1: Lipschitz continuity of f}
  \end{lemma}
  \begin{proof}
  For outputs $f(\bm{\lambda})$ and ${f}^{\prime}(\bm{\lambda})$, respectively we consider the following parameter sets: (a)  $\mathbf{W}^{(1)}_1, \cdots, \mathbf{W}^{(T)}_1, \mathbf{W}^{(1)}_2, \cdots, \mathbf{W}^{(T)}_2, \mathbf{W}_3, \mathbf{W}_4$, and (b) 
  $\mathbf{W}^{\prime(1)}_1, \cdots, \mathbf{W}^{\prime(T)}_1, \mathbf{W}^{\prime(1)}_2, \cdots, \mathbf{W}^{\prime(T)}_2, \mathbf{W}^{\prime}_3, \mathbf{W}^{\prime}_4$.
  For the $j$-{th} output in the NBP decoder, we have that, 
\begin{align}
\left\|f(\bm{\lambda}){[j]} \right. & - \left.f^{\prime}(\bm{\lambda}){[j]}\right\|_{2} =\left\|s\left(\mathbf{W}_4[j,:] \bm{\lambda}[j] + \mathbf{W}_3[j,:] \mathbf{p_{T}}\right)-s\left(\mathbf{W}^{\prime}_4[j,:] \bm{\lambda}[j] +\mathbf{W}^{\prime}_3[j,:] \mathbf{p^{\prime}_{T}}\right)\right\|_{2} \nonumber \\
& \leq \left\| \left(\mathbf{W}_4[j,:] -\mathbf{W}^{\prime}_4[j,:]\right) \bm{\lambda}[j] + \mathbf{W}_3[j,:] \mathbf{p_{T}} -  \mathbf{W}_3^{\prime}[j,:] \mathbf{p_{T}} +  \mathbf{W}_3^{\prime}[j,:] \mathbf{p_{T}} - \mathbf{W}_3^{\prime}[j,:] \mathbf{p^{\prime}_{T}} \right\|_{2} \nonumber  \\
& \leq \left\|\left(\mathbf{W}_4[j,:]- \mathbf{W}^{\prime}_4[j,:] \right) \bm{\lambda}[j]\right\|_{2} + \left\|\left(\mathbf{W}_3[j,:]- \mathbf{W}^{\prime}_3[j,:] \right) \mathbf{p_{T}}\right\|_{2}+\left\|\mathbf{W}^{\prime}_3[j,:]\left(\mathbf{p_{T}}-\mathbf{p^{\prime}_{T}}\right)\right\|_{2} \nonumber  \\
& \leq \left\|\mathbf{W}_4[j,:]- \mathbf{W}^{\prime}_4[j,:]  \right\|_{2} b_{\lambda} + \left\|\mathbf{p_{T}}\right\|_{2}\left\|\mathbf{W}_3[j,:]-\mathbf{W}^{\prime}_3[j,:]\right\|_{2}+B_{w_3}\left\|\mathbf{p_{T}}-\mathbf{p^{\prime}_{T}}\right\|_{2}.
\label{eq: y-lipschitz}
\end{align}    
where, $\left\|\mathbf{W}^{\prime}_3[j,:]\right\|_{2} \leq B_{w_3}$. 
We next find an upper bound $\left\|\mathbf{p_{T}}\right\|_{2}$ as a function of the number of iterations $T$, Lipschitz constants of the activation functions, and spectral norm bounds of the weight matrices of the NBP decoder. 
To further upper bound $\left\|\mathbf{p_{T}}\right\|_{2}$, we know from Lemma \ref{lemma 5: p_t leq v_t} that $\left\|\mathbf{p_{T}}\right\|_{2} \leq \left\|\mathbf{v_{T}}\right\|_{2}$, and therefore we have
\begin{align}
\left\|\mathbf{p_{T}}\right\|_{2} \leq \left\|\mathbf{v_{T}}\right\|_{2} &= \left\| \mathbf{W}^{(T)}_1 \bm{\lambda} + \mathbf{W}^{(T)}_2 \mathbf{p_{T-1}} \right\|_{2}  \nonumber \\
& \leq \left\| \mathbf{W}^{(T)}_1 \bm{\lambda} \right\|_{2}+\left\| \mathbf{W}^{(T)}_2 \mathbf{p_{T-1}}  \right\|_{2} \nonumber\\
& \leq \sqrt{n} B_{W_1} b_{\lambda}+B_{W_2}\left\|\mathbf{p_{T-1}}\right\|_{2}.
\label{eq: v_t-function-of-v_t-1}
\end{align}
where, (a) follows from Talagrand's concentration lemma \cite{ledoux2011probability}.
Applying \eqref{eq: v_t-function-of-v_t-1} recursively across $T$ decoding iterations we have, 
\begin{align}
\left\|\mathbf{p_{T}}\right\|_{2} \leq \left\|\mathbf{v_{T}}\right\|_{2}  \leq \sqrt{n}B_{W_1} b_{\lambda} \sum_{i=0}^{T-2} \left(B_{W_2}\right)^{i} + \left(B_{W_2}\right)^{T-1} \left\|\mathbf{v_{1}}\right\|_{2}.
\label{eq: v_t-function-of-v_1}
\end{align}
\noindent Also, we have $\mathbf{v_{1}} =  \mathbf{W}^{(1)}_1 \bm{\lambda} $.
Therefore, $\left\|\mathbf{v_{1}}\right\|_{2}$ can be upper bounded as
\begin{align}
\left\|\mathbf{v_{1}}\right\|_{2} = \left\|\mathbf{W}^{(1)}_1 \bm{\lambda} \right\|_{2}\leq \sqrt{n} B_{W_1} b_{\lambda}.
\label{eq: v_1-function-of-x}
\end{align}
Substituting \eqref{eq: v_1-function-of-x} in \eqref{eq: v_t-function-of-v_1}, we have
\begin{align}
\hspace{-0.25cm}\left\|\mathbf{p_{T}}\right\|_{2}
& \leq  \sqrt{n}B_{W_1} b_{\lambda} \left( \frac{\left( B_{W_2}\right)^{T-1} -1}{ B_{W_2} - 1} +   \left(B_{W_2}\right)^{T-1} \right).
\label{eq: p_t}
\end{align}

\noindent To upper bound $\left\|\mathbf{p_{T}}-\mathbf{p^{\prime}_{T}}\right\|_{2}$, we express $\left\|\mathbf{p_{T}} - \mathbf{p_{T}^{\prime}}\right\|_{2}$ in terms of $\left\|\mathbf{v_{T}}- \mathbf{v_{T}^{\prime}}\right\|_{2}$ as follows:
\begin{align}
    \left\|\mathbf{p_{T}} - \mathbf{p_{T}^{\prime}}\right\|_{2} \leq \sqrt{n} \left\|\mathbf{v_{T}}- \mathbf{v_{T}^{\prime}}\right\|_{2},
    \label{eq: p_tprime-function-of-p_t-1prime}
\end{align}
where, $\left\|\mathbf{v_{T}}- \mathbf{v_{T}^{\prime}}\right\|_{2}$ can be upper bounded as follows:
\begin{align}
 \left\|\mathbf{v_{T}}-  \mathbf{v_{T}^{\prime}}\right\|_{2} 
&= \left\|  \mathbf{W}^{(T)}_1 \bm{\lambda} +  \mathbf{W}^{(T)}_2 \mathbf{p_{T-1}} - \left(\mathbf{W}_1^{\prime(T)} \bm{\lambda} +  \mathbf{W}^{\prime(T)}_2 \mathbf{p_{T-1}^{\prime}}\right)\right\|_{2} \nonumber \\
%
%
& \leq  \left\|\left(\mathbf{W}^{(T)}_1-\mathbf{W}_1^{\prime{(T)}}\right) \bm{\lambda} \right\|_{2} + \left\|\mathbf{W}^{(T)}_2 \mathbf{p_{T-1}}-\mathbf{W}^{\prime (T)}_2\mathbf{p_{T-1}^{\prime}}\right\|_{2} \nonumber \\
& \leq \sqrt{n} b_{\lambda}\left\|\mathbf{W}^{(T)}_1 - \mathbf{W}_1^{\prime{(T)}}\right\|_{2}   \hspace{-0.05cm}+\hspace{-0.05cm}  \left\|\mathbf{W}^{(T)}_2 \mathbf{p_{T-1}} \hspace{-0.1cm}-\hspace{-0.1cm}\mathbf{W}^{\prime {(T)}}_2 \mathbf{p_{T-1}} \hspace{-0.1cm}+ \hspace{-0.1cm}\mathbf{W}^{\prime {(T)}}_2 \mathbf{p_{T-1}} \hspace{-0.1cm}-\hspace{-0.1cm} \mathbf{W}^{\prime {(T)}}_2 \mathbf{p_{T-1}^{\prime}}\right\|_{2} \nonumber \\
& \leq  \sqrt{n} b_{\lambda} \left\|\mathbf{W}^{(T)}_1-\mathbf{W}_1^{\prime {(T)}} \right\|_{2} \hspace{-0.05cm}+ \hspace{-0.05cm} \left\|\mathbf{p_{T-1}}\right\|_{2}\left\|\mathbf{W}^{(T)}_2 -\mathbf{W}_2^{\prime {(T)}}\right\|_{2} \hspace{-0.05cm} + \hspace{-0.05cm} B_{W_2}\left\|\mathbf{p_{T-1}}-\mathbf{p^{\prime}_{T-1}}\right\|_{2} \nonumber \\
& \leq  \sqrt{n} b_{\lambda} \left\|\mathbf{W}^{(T)}_1 \hspace{-0.05cm} - \hspace{-0.05cm} \mathbf{W}_1^{\prime{(T)}}\right\|_{2}\hspace{-0.05cm}+\hspace{-0.05cm}\left\|\mathbf{p_{T-1}}\right\|_{2}\left\|\mathbf{W}^{(T)}_2 \hspace{-0.05cm} - \hspace{-0.05cm} \mathbf{W}^{\prime{(T)}}_2 \right\|_{2} \hspace{-0.05cm}+\hspace{-0.05cm}\sqrt{n}B_{W_2}\left\|\mathbf{v_{T-1}} \hspace{-0.05cm} - \hspace{-0.05cm}  \mathbf{v_{T-1}^{\prime}}\right\|_{2}.
\label{eq: v_tprime-function-of-v_t-1prime}
\end{align}

\noindent Applying \eqref{eq: v_tprime-function-of-v_t-1prime} recursively across $T$ decoding iterations we have,
\begin{align}
\left\|\mathbf{v_{T}}-\mathbf{v_{T}^{\prime}}\right\|_{2} 
& \leq   \sqrt{n} b_{\lambda} \left(\sum^{T-2}_{i=0} \left(\sqrt{n} B_{W_2} \right)^{i} \left\|\mathbf{W}^{(T-i)}_1 -\mathbf{W}_1^{\prime {(T-i)}}\right\|_{2} \right) \nonumber \\
& + \left(\sum_{i=1}^{T-1} \left( \sqrt{n}  B_{W_2} \right)^{T-1-i} \hspace{-0.05cm} \left\|\mathbf{v_{i}}\right\|_{2} \left\|\mathbf{W}^{(i+1)}_2 \hspace{-0.1cm} - \hspace{-0.1cm} \mathbf{W}_2^{\prime (i+1)}\right\|_{2} \right)  
 + \left( \sqrt{n} B_{W_2}\right)^{T-1} \hspace{-0.1cm} \left\|\mathbf{v_{1}}- \mathbf{v_{1}^{\prime}}\right\|_{2}
  \label{eq: v_tprime-function-of-v_1prime}
\end{align}

\noindent To upper bound $\left\|\mathbf{v_{1}}- \mathbf{v_{1}^{\prime}}\right\|_{2}$, we have that,
\begin{align}
\left\|\mathbf{v_{1}}- \mathbf{v_{1}^{\prime}}\right\|_{2} 
& = \left\|\mathbf{W}^{(1)}_1 \bm{\lambda} - \mathbf{W}^{\prime{(1)}}_1  \bm{\lambda} \right\|_{2}
\nonumber \\
& \leq  \left\|\bm{\lambda} \right\|_{2} \left\| \mathbf{W}^{(1)}_1  - \mathbf{W}^{\prime{(1)}}_1 \right\|_{2}  .
 \label{eq: v_tprime-function-of-x}
\end{align}

\noindent In addition, the term $\sum_{i=1}^{T-1} \left( \sqrt{n}  B_{W_2} \right)^{T-1-i}  \left\|\mathbf{v_{i}}\right\|_{2}$ can be upper bounded as follows,
\begin{align}
     \sum_{i=1}^{T-1} \left( \sqrt{n}  B_{W_2} \right)^{T-1-i}  & \left\|\mathbf{v_{i}}\right\|_{2}  \leq \sum_{i=1}^{T-1} \left(\sqrt{n}  B_{W_2} \right)^{T-1-i} \hspace{-0.1cm} \times  \left((i-1) \sqrt{n} B_{W_1} b_{\lambda} +\sqrt{n} b_{\lambda} B_{W_1} \left(B_{W_2}\right)^{i-1}  \right) \nonumber \\ 
    & \hspace{-0.7cm} \leq T \sqrt{n} B_{W_1} b_{\lambda}  \frac{\left(\sqrt{n}  B_{W_2} \right)^{T-1}-1}{\sqrt{n}  B_{W_2}-1} + \sqrt{n} b_{\lambda}   B_{W_1}  (B_{W_2})^{T-2}     \sum_{i=1}^{T-1} (\sqrt{n})^{T-1-i} \nonumber \\ 
    & \hspace{-0.7cm} \leq T \sqrt{n} B_{W_1} b_{\lambda} \frac{\left(\sqrt{n}  B_{W_2} \right)^{T-1}-1}{\sqrt{n} B_{W_2}-1}   +  \sqrt{n} b_{\lambda}   B_{W_1}   (B_{W_2})^{T-2}  \frac{(\sqrt{n})^{T-1}-1}{\sqrt{n}-1}.
     \label{eq: v_i-upperbound}
\end{align}

\noindent Substituting \eqref{eq: v_tprime-function-of-x} and \eqref{eq: v_i-upperbound} in \eqref{eq: v_tprime-function-of-v_1prime}, we have
\begin{align}
& \left\|\mathbf{v_{T}}-\mathbf{v_{T}^{\prime}}\right\|_{2} 
\leq  \sqrt{n} b_{\lambda} \left(\sum^{T-1}_{i=0} \left(\sqrt{n} B_{W_2} \right)^{i} \left\|\mathbf{W}^{(T-i)}_1 -\mathbf{W}_1^{\prime {(T-i)}}\right\|_{2} \right) \nonumber \\
& +  \left( T \sqrt{n} b_{\lambda} B_{W_1} \frac{\left(\sqrt{n}  B_{W_2} \right)^{T-1}-1}{\sqrt{n}  B_{W_2}-1} + \sqrt{n} b_{\lambda}   B_{W_1}    {\hspace{-0.05cm} \left(\hspace{-0.05cm}B_{W_2}\hspace{-0.05cm}\right)\hspace{-0.05cm}}^{T-2}     \frac{(\sqrt{n})^{T-1}-1}{\sqrt{n}-1}\right) \sum^{T}_{i=2}  \left\|\mathbf{W}^{(i)}_2-\mathbf{W}^{\prime{(i)}}_2\right\|_{2}
\label{eq: v_tprime-function-of-uvw}
\end{align}

\noindent Substituting \eqref{eq: v_tprime-function-of-uvw} in \eqref{eq: p_tprime-function-of-p_t-1prime} we have that, 
\begin{align}
     & \left\|\mathbf{p_{T}} - \mathbf{p_{T}^{\prime}}\right\|_{2} \leq   n b_{\lambda} \left(\sum^{T-1}_{i=0} \left(\sqrt{n} B_{W_2} \right)^{i} \left\|\mathbf{W}^{(T-i)}_1 -\mathbf{W}_1^{\prime {(T-i)}}\right\|_{2} \right)
\nonumber \\
& + \left( T n b_{\lambda} B_{W_1} \frac{\left(\sqrt{n}  B_{W_2} \right)^{T-1}-1}{\sqrt{n}  B_{W_2}-1} +  n b_{\lambda}   B_{W_1}    \left(B_{W_2}\right)^{T-2}   \frac{(\sqrt{n})^{T-1}-1}{\sqrt{n}-1}\right)  \sum^{T}_{i=2} \left\|\mathbf{W}^{(i)}_2-\mathbf{W}^{\prime (i)}_2\right\|_{2}
\label{eq: v_tprime-function-of-uvw-2}
\end{align}

\noindent Using \eqref{eq: p_t} and \eqref{eq: v_tprime-function-of-uvw-2} in \eqref{eq: y-lipschitz} we obtain
\begin{align}
& \left\|f(\bm{\lambda}){[j]} - f^{\prime}(\bm{\lambda}){[j]}\right\|_{2} 
 \leq   n  b_{\lambda} B_{w_3} \left(\sum^{T-1}_{i=0} \left(\sqrt{n} B_{W_2} \right)^{i} \left\|\mathbf{W}^{(T-i)}_1 -\mathbf{W}_1^{\prime {(T-i)}}\right\|_{2} \right) \nonumber \\
& +   n b_{\lambda} B_{W_1} B_{w_3} \left(  T \frac{\left(\sqrt{n}  B_{W_2} \right)^{T-1}-1}{\sqrt{n}  B_{W_2}-1} + \left(B_{W_2}\right)^{T-2} \frac{(\sqrt{n})^{T-1}-1}{\sqrt{n}-1}\right)  \sum^{T}_{i=2} \left\|\mathbf{W}^{(i)}_2-\mathbf{W}^{\prime (i)}_2 \right\|_{2}
\nonumber \\ &  +   \sqrt{n}B_{W_1} b_{\lambda} \left( \frac{\left( B_{W_2}\right)^{T-1}-1}{ B_{W_2} - 1} +   \left(B_{W_2}\right)^{T-1} \right)  \times \left\|\mathbf{W}_3[j,:]-\mathbf{W}^{\prime}_3[j,:]\right\|_{2}\nonumber \\ & +  b_{\lambda}\left\|\mathbf{W}_4[j,:]- \mathbf{W}^{\prime}_4[j,:]  \right\|_{2}.
\end{align} 

\noindent This completes the proof of Lemma \ref{Lemma-1: Lipschitz continuity of f}.
  \end{proof}
  \noindent Next, we state and prove Lemma \ref{lemma 5: p_t leq v_t} in which we show that for any layer $t$, the value of $\|\mathbf{p_t}\|_2$ is less than $\|\mathbf{v_t}\|_2$. This inequality was used to obtain the result in Lemma \ref{lemma-1: f is lipschitz}. 
  \begin{lemma}
Given the output of the parity check layer $\mathbf{p_t}$ defined by the min-sum operation as:
\begin{align}
    \hspace{-0.22cm} \mathbf{p_t}[\{l,m\}] = \hspace{-0.5cm} \underset{l^{\prime} \in \mathcal{P}(m)\backslash l}{\prod} \hspace{-0.3cm}  sign({\mathbf{v_t}[\{l^{\prime} ,m\}]}) \underset{l^{\prime} \in \mathcal{P}(m)\backslash l}{\min} |{\mathbf{v_t}[\{l^{\prime} ,m\}]}|
\end{align}
\noindent Then, the norm $\|\mathbf{p_t}\|_2$ is always bounded by the norm $\|\mathbf{v_t}\|_2$, i.e., $\|\mathbf{p_t}\|_2 \leq \|\mathbf{v_t}\|_2$. 
\label{lemma 5: p_t leq v_t}
\end{lemma}
\begin{proof}
It is straightforward to verify that for any decoding iteration the norm of output of parity check layer is always lower than the variable node layer.
For $n$ length code and $k$ length message, let us consider parity check matrix $\mathbf{H} \in \{0,1\}^{n \times (n-k)}$ 
with variable node degree $d_v$, and parity check node degree $d_c$.

\noindent We denote the corresponding Tanner graph as $\mathcal{G} \in \{\mathcal{V},\mathcal{P}, \mathcal{E} \}$, where $\mathcal{V} = \{v_1,\dots,v_n \}$, $\mathcal{P} = \{p_1,\dots,p_{n-k} \}$, and $\mathcal{E} = \{e_1,\dots,e_{nd_v} \}$.
Without loss of generality, let us consider parity check node $p_1$ in $\mathcal{G}$ such that $\mathcal{P}(p_1) = \{v_1, v_2, \dots, v_{d_c} \}$, where $\{v_1, v_2, \dots, v_{d_c} \} \subseteq \mathcal{V}$.
Therefore, in the Tanner graph, $d_c$ is the number of incoming messages to $p_1$; which translates to $d_c$ hidden nodes in the variable check layer $\mathbf{v_t}$ for any $1 \leq  t \leq T$ in the NBP decoder (whose architecture was described in Section \ref{sec: main-results}). 
Similarly, $d_c$ is the number of outgoing messages from $p_1$; this translates to $d_c$ hidden nodes in the parity check layer $\mathbf{p_t}$ in the NBP decoder.

\noindent The message passed $\mathbf{p_t} [\{1, 1\}]$ from $p_1$ to $v_1$ in the NBP decoder is given as, 
\begin{align}
   \mathbf{p_t} [\{1, 1\}] & =   \underset{l^{\prime} \in \mathcal{P}(p_1)\backslash v_1}{\prod} \hspace{-0.3cm}  sign({\mathbf{v_t}[\{l^{\prime} ,1\}]}) \underset{l^{\prime} \in \mathcal{P}(p_1)\backslash v_1}{\min} |{\mathbf{v_t}[\{l^{\prime} ,1\}]}| \nonumber \\
     =  & sign({\mathbf{v_t}[\{{2} ,1\}]}) \cdot sign({\mathbf{v_t}[\{{3} ,1\}]}) \cdots sign({\mathbf{v_t}[\{{d_c} ,1\}]})  \nonumber \\ & \hspace{0.25cm}\times \min\left(|{\mathbf{v_t}[\{{2} ,1\}]}|  ,|{\mathbf{v_t}[\{{3} ,1\}]}| ,\cdots ,|{\mathbf{v_t}[\{{d_c} ,1\}]}|\right) \nonumber \\
     \stackrel{(a)}{=}  & sign({\mathbf{v_t}[\{{2} ,1\}]}) \cdot sign({\mathbf{v_t}[\{{3} ,1\}]}) \cdots sign({\mathbf{v_t}[\{{d_c} ,1\}]})  \times |{\mathbf{v_t}[\{{2} ,1\}]}| .
     \label{eq: lemma-2-eq-1}
\end{align}

\noindent 
where, step (a) follows from our assumption that $|\mathbf{v_t} [\{1,1\}]| < |\mathbf{v_t} [\{{2}, 1\}]| \cdots < |\mathbf{v_t} [\{{d_c}, 1\}]|$, without loss of generality. 
Following similar steps as in \eqref{eq: lemma-2-eq-1}, for the set $\{ v_i | 2 \leq i \leq d_c\} $, we have that,
\begin{align}
   \mathbf{p_t} [\{i, 1\}] = \hspace{-0.5cm}  \underset{l^{\prime} \in \mathcal{P}(p_1)\backslash v_i}{\prod} \hspace{-0.3cm}  sign({\mathbf{v_t}[\{l^{\prime} ,1\}]}) \times |{\mathbf{v_t}[\{1 ,1\}]}|
\end{align}

\noindent This implies that outputs of $d_c - 1$ nodes in $\mathbf{p_t}$ correspond to the minimum absolute value of the $d_c$ incoming messages to $p_1$. 
Therefore, for the parity check equation $p_1$, we have, 
\begin{align}
    \left( \sum_{i = 1}^{d_c} |\mathbf{p_t} [\{i,1\}]|^2\right)^{\frac{1}{2}}  &\leq 
    \left((d_c - 1) |\mathbf{v_t} [\{1,1\}]|^2 + |\mathbf{v_t} [\{{2}, 1\}]|^2\right)^{\frac{1}{2}}.
\end{align}
Following similar steps for parity checks $p_2, \cdots p_{n-k}$, we can conclude that,
\begin{align}
    \|\mathbf{p_t}\|_2 \leq \|\mathbf{v_t}\|_2,
    \label{eq: lemma-2-conclusion}
\end{align}
where, equality in \eqref{eq: lemma-2-conclusion} is satisfied when the output of the hidden nodes in $\mathbf{v_t}$ have same absolute values. 
This completes the proof of Lemma \ref{lemma 5: p_t leq v_t}.
\end{proof}

\section{Bound on covering number of sparse matrices}
\label{sec: Lemma-3-Proof}
\noindent In this section, we derive an upper bound on the covering number of sparse matrices as a function of its rows, columns, and spectral norm bound.   
  \begin{lemma} Let $\mathcal{W} = \{\mathbf{W} \in \mathbb{R}^{r \times c}: \|\mathbf{W}\|_2 \leq B_W \text{ and }\|\mathbf{W}[i,:]\|_0 = q ,\text{ }1 \leq i \leq r \}$ be the  space of matrices with its spectral norm bounded by a constant $B_W$, and exactly $q$ non-zero entries in each row.
  Then, its  covering number $\mathcal{N}(\mathbf{W}, \epsilon, \|\cdot\|_F)$ with respect to the Frobenius norm can be upper bounded as follows:  
  \begin{align}
          \mathcal{N}(\mathcal{W}, \epsilon, \|\cdot\|_F) \leq {\left(1 + \frac{2\min{(\sqrt{r},\sqrt{c})} B_W}{\epsilon}\right)^{qr}}.
  \end{align}
  where, $\epsilon > 0$ is a constant. 
  
  \label{Lemma-2: Covering number matrices}
  \end{lemma}

\begin{proof}
We consider $\Psi: \mathbb{R}^{r \times c} \rightarrow \mathbb{R}^{qr \times 1}$, a bijective mapping such that $\Psi(\mathbf{W}) \in \mathbb{R}^{q r  \times 1}$ is a vector of all non-zero entries in $\mathbf{W} \in \mathcal{W}$. 
The vector space induced by the mapping $\Psi$ is defined as $\mathcal{\psi}$ such that $\psi(\mathcal{W}) = \{\Psi(\mathbf{W}): \mathbf{W} \in \mathcal{W} \}$.
Therefore, we also have $\|\mathbf{W}\|_{F} = \|\Psi(\mathbf{W})\|_{2}$.

 \noindent We construct $\mathcal{C}(\psi(\mathcal{W}), \epsilon, \|\cdot\|_2)$, a $\epsilon$-covering of $\psi(\mathcal{W})$ under the $L_2$ norm, and denote the corresponding covering number by $\mathcal{N}(\psi(\mathcal{W}), \epsilon, \|\cdot\|_2)$.
Similarly, we construct $\mathcal{C}(\mathcal{W}, \epsilon, \|\cdot\|_F)$, a $\epsilon$-covering of $\mathcal{W}$ under $\|\cdot\|_F$, and its covering number is denoted as $\mathcal{N}(\mathcal{W}, \epsilon, \|\cdot\|_F)$.  
In what follows, we have,
\begin{align}
    \mathcal{N}(\mathcal{W}, \epsilon, \|\cdot\|_F) \stackrel{(a)}{\leq} \mathcal{N}(\psi(\mathcal{W}), \epsilon, \|\cdot\|_2) \stackrel{(b)}{\leq}  \mathcal{M}(\psi(\mathcal{W}), \epsilon, \|\cdot\|_2).
    \label{eq: covering-no-matrix-vector}
\end{align}
The inequality (a) in \eqref{eq: covering-no-matrix-vector} follows from the fact that $\Psi$ is a bijective map, and that $\Psi^{-1} \mathcal{C}(\psi(\mathcal{W}), \epsilon, \|\cdot\|_2)$ is also a $\epsilon$-cover of the matrices in $\mathcal{W}$.
The inequality (b) in \eqref{eq: covering-no-matrix-vector} follows from the definition of packing $\mathcal{P}(\psi(\mathcal{W}), \epsilon, \|\cdot\|_2)$ and the packing number $\mathcal{M}(\psi(\mathcal{W}), \epsilon, \|\cdot\|_2)$. 
In other words, let $\mathcal{P}(\psi(\mathcal{W}), \epsilon, \|\cdot\|_2)$ be the maximal packing.
Suppose for some $\mathbf{W}_i \in \mathcal{W}$ there exists $\Psi(\mathbf{U}) \in \mathcal{W}\backslash \mathcal{P}(\psi(\mathcal{W}), \epsilon, \|\cdot\|_2)$ such that,
\begin{align}
    \|\Psi(\mathbf{W}_i) - \Psi(\mathbf{U})\|_2 \geq \epsilon.
\end{align}
Then, $\mathcal{P}(\psi(\mathcal{W}), \epsilon, \|\cdot\|_2)$ being a maximal packing is a contradiction. 
Therefore, $\mathcal{P}(\psi(\mathcal{W}), \epsilon, \|\cdot\|_2)$ is also an $\epsilon$-cover of $\psi(\mathcal{W})$ and it follows that its cardinality (i.e., packing number) denoted by  $\mathcal{M}(\psi(\mathcal{W}), \epsilon, \|\cdot\|_2)$ is greater than $\mathcal{N}(\psi(\mathcal{W}), \epsilon, \|\cdot\|_2)$.
 In the next step, we upper bound $\mathcal{M}(\psi(\mathcal{W}), \epsilon, \|\cdot\|_2)$.
To this end, we make use of the definitions \ref{definition-1: covering-number}, \ref{definition-2: packing-number} to determine this upper bound; and it follows from these definitions that the balls need not be disjoint for an $\epsilon$-cover $\mathcal{C}(\psi(\mathcal{W}), \epsilon, \|\cdot\|_2)$, and it must be disjoint for $\epsilon$-packing $\mathcal{P}(\psi(\mathcal{W}), \epsilon, \|\cdot\|_2)$. Therefore, $\forall \mathbf{w}_i \in \mathcal{P}(\psi(\mathcal{W}), \epsilon, \|\cdot\|_2)$ we have that, 
\begin{align}
    \bigcup_{i=1}^{\mathcal{P}(\psi(\mathcal{W}), \epsilon, \|\cdot\|_2)} \mathcal{B} (\mathbf{w}_i, \frac{\epsilon}{2}) \subset \mathcal{B} (0, R + \frac{\epsilon}{2}).
    \label{eq: upper-bound-packing-eq0}
\end{align}
where, the radius $R =  \underset{\mathbf{W} \in \mathcal{W}}{\max} \|\Psi(\mathbf{W})\|_{2}$. Taking volume on both sides in \eqref{eq: upper-bound-packing-eq0}, we have that, 
\begin{align}
    & \text{vol} \left( \bigcup_{i=1}^{\mathcal{P}(\psi(\mathcal{W}), \epsilon, \|\cdot\|_2)} \mathcal{B} (\mathbf{w}_i, \frac{\epsilon}{2}) \right) \leq \text{vol} \left(\mathcal{B} (0, R + \frac{\epsilon}{2})\right) \nonumber \\
    \implies & 
    \hspace{-0.5cm}\sum_{i=1}^{\mathcal{P}(\psi(\mathcal{W}), \epsilon, \|\cdot\|_2)}\text{vol} \left(  \mathcal{B} (\mathbf{w}_i, \frac{\epsilon}{2}) \right) \leq \text{vol} \left(\mathcal{B} (0, R + \frac{\epsilon}{2})\right).
    \label{eq: upper-bound-packing-eq1}
\end{align}
To find the value of the radius $R$, we bound the $L_2$ norm of $\Psi(\mathbf{W})$ in terms of the spectral norm of the sparse matrix $\mathbf{W}$, and we have that, 
\begin{align}
     \|\Psi(\mathbf{W})\|_{2} = \|\mathbf{W}\|_{F} & \leq  \min{(\sqrt{r},\sqrt{c})} \|\mathbf{W}\|_2 \leq  \min{(\sqrt{r},\sqrt{c})} B_W.
\end{align}
By considering the $L_2$ ball $\mathcal{B}(0,R)$, where radius $R = \min{(\sqrt{r},\sqrt{c})} B_W$, and from \eqref{eq: upper-bound-packing-eq1} we obtain an upper bound on $\mathcal{M}(\psi(\mathcal{W}), \epsilon, \|\cdot\|_2)$ as follows,
\begin{align}
    \mathcal{M}(\psi(\mathcal{W}), \epsilon, \|\cdot\|_2) & \leq \frac{\left(R + \frac{\epsilon}{2}\right)^{qr}}{\left(\frac{\epsilon}{2}\right)^{qr}}
    \leq {\left(1 + \frac{2\min{(\sqrt{r},\sqrt{c})} B_W}{\epsilon}\right)^{qr}}.
\end{align}
Therefore, we have, 
\begin{align}
    \mathcal{N}(\mathcal{W}, \epsilon, \|\cdot\|_F) &{\leq} \mathcal{M}(\psi(\mathcal{W}), \epsilon, \|\cdot\|_2) \leq {\left(1 + \frac{2\min{(\sqrt{r},\sqrt{c})} B_W}{\epsilon}\right)^{qr}}. 
\end{align}

\noindent This completes the proof of Lemma \ref{Lemma-2: Covering number matrices}.
\end{proof}

\section{Proof of Theorem \ref{Corollary-1: irregular-parity-check-mat}}
\label{sec: Corollary-1-Proof}

\noindent We consider an irregular parity check matrix $\mathbf{H} \in \mathbb{R}^{(n-k) \times n}$  where $d_{v_i}$ is the variable node degree of the $i$-{th} bit in the codeword, and $d_{c_j}$ is the parity check node degree of the $j$-{th} parity check equation.
The NBP decoder corresponding to such this irregular parity check matrix is characterized by the weight matrices $\mathbf{W}^{(1)}_1$, $\mathbf{W}^{(2)}_1, \cdots, \mathbf{W}^{(T)}_1$, $\mathbf{W}^{(2)}_2$, $\mathbf{W}^{(2)}_2, \cdots, \mathbf{W}^{(T)}_2$, $\mathbf{W}_3$, $\mathbf{W}_4$.

\noindent Here, for $t \in \{1,\cdots, T\}$, and $\theta = \sum\limits_{i = 1}^{n} d_{v_i}$, we have that $\mathbf{W}^{(t)}_1 \in \mathbb{R}^{ \theta \times n}$, $\mathbf{W}^{(t)}_2 \in \mathbb{R}^{ \theta \times \theta}$, $\mathbf{W}_3 \in \mathbb{R}^{ n \times \theta}$, and 
$\mathbf{W}_4 \in \mathbb{R}^{ n \times n}$.
For any value of $t$, the weight matrix $\mathbf{W}^{(t)}_1$ has one non-zero entry in every row, and $d_{v_i}$ non-zero entries in the $i$-th column for integer values $i \in [n]$. 
In the weight matrix $\mathbf{W}^{(t)}_2$, the $i$-th bit in the codeword with variable node degree $d_{v_i}$ corresponds to $d_{v_i}$ rows and $d_{v_i}$ columns, and these rows and columns each have exactly $d_{v_{i}}-1$ non-zero entries.

\noindent The $i$-{th} row in the weight matrix $\mathbf{W}_3$ corresponds to the $i$-{th} bit in the codeword, and has exactly $d_{v_i}$ non-zero entries.
Lastly, the weight matrix $\mathbf{W}_4 \in \mathbb{R}^{n \times n}$ is a a diagonal matrix.
 If the weights in the NBP decoder are bounded in $[-w,w]$, then for any $t \in \{1,2, \cdots, T\}$, the spectral norm of the weight matrices $\mathbf{W}^{(t)}_1$, and $\mathbf{W}^{(t)}_2$ can be bounded as follows: 
\begin{align}
   & B_{W_1} = \|\mathbf{W}^{(t)}_1\|_2 \leq w \sqrt{ \mathop{\max}_{i} d_{v_i}} ,   \hspace{1.8cm}
    B_{W_2} = \|\mathbf{W}^{(t)}_2\|_2 \leq w \left(\mathop{\max}_{i} d_{v_i} - 1 \right) 
   \label{eq: spectral norm bounds w1-w2-irregular}
\end{align}
The $L_2$ norm bounds on the row vector in matrices $\mathbf{W}_3$ and $\mathbf{W}_4$ are as follows:
\begin{align}
   & B_{w_3} = \|\mathbf{W}_3[j,:]\|_2 \leq w \sqrt{d_{v_j}} \leq w \sqrt{ \mathop{\max}_{i} d_{v_i}},   \hspace{1cm}
    B_{w_4} = \|\mathbf{W}_4[j,:]\|_2 \leq w
   \label{eq: spectral norm bounds w3-w4-irregular}
\end{align}

\noindent For irregular parity check matrices, the upper bound \eqref{eq: covering-number-matrices-prod} becomes,
\begin{align}
     \mathcal{N}(\mathcal{F}_T[j], {\epsilon}, \|\cdot\|_2)
     & \leq 
      {\prod_{i = 1}^{T} \left(1 + \frac{ (4T + 2) \sqrt{n} B_{W_1}\rho_{_{W^{(i)}_1}}}{\epsilon}\right)^{\mathop{\sum}\limits_{h=1}^{n} d_{v_h}}}   \nonumber \\&  \times  {\prod_{i = 2}^{T} \left(1 + \frac{ (4T + 2) \sqrt{nd_v} B_{W_2}\rho_{_{W^{(i)}_2}}}{\epsilon}\right)^{\mathop{\sum}\limits_{h=1}^{n} (d_{v_h} - 1)d_{v_h}}} 
     \nonumber \\& \times {\left(1 + \frac{ (4T + 2)  B_{w_3}\rho_{_{w_3}}}{\epsilon}\right)^{\max\limits_{h} d_{v_h}}}  \times  {\left(1 + \frac{ (4T + 2) B_{w_4}\rho_{_{w_4}}}{\epsilon}\right)}. 
     \label{eq: covering-number-matrices-prod-irregular}
\end{align}
 We can further upper bound  $\mathcal{N}(\mathcal{F}_T[j], {\epsilon}, \|\cdot\|_2)$ in  \eqref{eq: covering-number-matrices-prod-irregular} as follows.
\begin{align}
\mathcal{N}(\mathcal{F}_T[j], {\epsilon}, \|\cdot\|_2) \leq \left(1 + \frac{(4T + 2) \sqrt{nd_v}\left(c_1 + c_2 \right)}{\epsilon}  \right)^{(T + 1) \sum\limits_{h = 1}^{n} d^2_{v_h}}. 
\label{eq: covering-number-Ft-using-matrices-irregular}
\end{align}
where, $c_1 = nT B_{W_1} B_{W_2} B_{w_3} b_{\lambda} \frac{( \sqrt{n} B_{W_2})^{T-1}-1}{\sqrt{n} B_{W_2}-1}$,  $c_2 = n B_{W_1} B_{w_3} b_{\lambda} \left(\sqrt{n} B_{W_2} \right)^{T-1}$. 
Substituting the values of the spectral norm bounds in \eqref{eq: spectral norm bounds w1-w2-irregular} and \eqref{eq: spectral norm bounds w3-w4-irregular} in \eqref{eq: covering-number-Ft-using-matrices-irregular} and assuming $n >> T$, we obtain,
\begin{align}
\mathcal{N}(\mathcal{F}_T[j], {\epsilon}, \|\cdot\|_2) \leq \left(1 + \frac{ (4T+2)\sqrt{n \max\limits_{i} d_{v_i}}}{\epsilon} b_{\lambda} (T+1) \left(\sqrt{n}w\max\limits_{i} d_{v_i}\right)^{T+1} \right)^{(T + 1) \sum\limits_{h = 1}^{n} d^2_{v_h}}
\end{align}

\noindent For the NBP decoder corresponding to an irregular parity check matrices, the bit-wise Rademacher complexity is upper bounded as,
\begin{align}
    R_m(\mathcal{F}_{T}[j]) \leq   \frac{4}{{m}} + \frac{12}{m} \sqrt{{m(T + 1) \sum_{h = 1}^{n} d^2_{v_h}} \log \hspace{-0.1cm}\left(\hspace{-0.1cm} {8 \sqrt{mn \max\limits_{i} d_{v_i}}} b_{\lambda} (T+1)^2 \hspace{-0.1cm} \left(\hspace{-0.1cm}\sqrt{n}w \max\limits_{i} d_{v_i}\hspace{-0.1cm}\right)^{T+1} \right)\hspace{-0.05cm}}
\end{align}

\noindent Subsequently, we have the upper bound on the true risk $\mathcal{R}_\text{BER}(f)$ as,
\begin{align}
    \mathcal{R}_\text{BER}(f)
    \leq \mathcal{\hat{R}}_{\text{BER}}(f)+ {12} \sqrt{\frac{\sum\limits_{h = 1}^{n} d^2_{v_h}({T+1})^2}{m} \log \left(8{\sqrt{mn}w\max_{i} d_{v_i} b_{\lambda}}\right)}  + \frac{4}{{m}} +\sqrt{\frac{\log(1/\delta)}{2m}}.
\end{align}

\section{Proof of Theorem \ref{Proposition-2: unbounded-input-channel-snr}}
\label{sec: Proposition-2-Proof}

\noindent From law of total expectations, we have that,
 \begin{align}
  \operatorname{Pr}(l_\text{BER}&(f(\bm{\lambda}),\mathbf{x}) > 0)   = \mathop{\mathbb{E}}\left[l_\text{BER}(f(\bm{\lambda}),\mathbf{x}) \bigg{|} \forall i \in [n] \text{ s.t. } |\bm{\lambda}[i]| \leq b_{\lambda}\right]   \times \underbrace{\operatorname{Pr}\left(\forall i \in [n] \text{ s.t. } |\bm{\lambda}[i]| \leq b_{\lambda}\right)}_{\text {prob. that input is bounded}} \nonumber \\ &+ \mathop{\mathbb{E}}\left[l_\text{BER}(f(\bm{\lambda}),\mathbf{x}) \bigg{|} \exists i \in [n] \text{ s.t. } |\bm{\lambda}[i]| > b_{\lambda} \right]   \times \underbrace{\operatorname{Pr}\left(\exists i \in [n] \text{ s.t. } |\bm{\lambda}[i]| > b_{\lambda} \right)}_{\text {prob. that input is unbounded}}
 \end{align}

 \noindent We obtain the following inequality using the fact that any probability is upper bounded by 1 for the term $\operatorname{Pr}\left(\forall i \in [n] \text{ s.t. } |\bm{\lambda}[i]| \leq b_{\lambda}\right)$, and that the true risk conditioned on the event that the log-likelihood ratios are unbounded is also upper bounded by $1$.
 \begin{align}
  \operatorname{Pr}(l_\text{BER}(f(\bm{\lambda}),\mathbf{x}) > 0)   & \leq  \mathop{\mathbb{E}}\left[l_\text{BER}(f(\bm{\lambda}),\mathbf{x}) \bigg{|} \forall i \in [n] \text{ s.t. } |\bm{\lambda}[i]| \leq b_{\lambda}\right]  +  \operatorname{Pr}\left(\exists i \in [n] \text{ s.t. } |\bm{\lambda}[i]| > b_{\lambda} \right)
   \label{eq: total-probability}
 \end{align}
 
 \noindent Using the fact that the channel outputs are i.i.d to compute $\operatorname{Pr}\left(\forall i \in [n] \text{ s.t. } |\bm{\lambda}[i]| \leq b_{\lambda}\right)$ as,
 
 \begin{align}
     \operatorname{Pr}\left(\forall i \in [n] \text{ s.t. } |\bm{\lambda}[i]| \leq b_{\lambda}\right) = \prod_{i = 1}^{n} \operatorname{Pr}\left( |\bm{\lambda}[i]| \leq b_{\lambda}\right)
 \end{align}
 We consider that the signal is modulated by binary phase shift keying (BPSK) modulation such that $\operatorname{Pr}(+1) = \operatorname{Pr}(-1) = \frac{1}{2} $. 
 The channel is AWGN channel  with noise variance $\beta^2$, then $\bm{\lambda}[i] = 2 {\mathbf{y}[i]}/{\beta^2}$.
 We can upper bound the term $\operatorname{Pr}\left( |\mathbf{y}[i]| \leq \frac{\beta^2 b_{\lambda}}{2} \right)$ using Q-function as follows.
 \begin{align}
     \operatorname{Pr}\left( |\bm{\lambda}[i]| \leq b_{\lambda}\right) = \left(1 - Q\left(\frac{\beta^2 b_{\lambda}+ 2}{2\beta} \right) - Q\left(\frac{\beta^2 b_{\lambda}- 2}{2\beta} \right) \right),
     \label{eq: prob-input-is-bounded}
 \end{align}
Then, the term  $\operatorname{Pr}\left(\exists i \in [n] \text{ s.t. } |\bm{\lambda}[i]| > b_{\lambda} \right)$ is computed as,
 \begin{align}
     \operatorname{Pr}\left(\exists i \in [n] \text{ s.t. } |\bm{\lambda}[i]| > b_{\lambda} \right) = 1- \left(1 - Q\left(\frac{\beta^2 b_{\lambda}+ 2}{2\beta} \right) - Q\left(\frac{\beta^2 b_{\lambda}- 2}{2\beta} \right)\right)^n
     \label{eq: prob-input-is-unbounded}
 \end{align}
 Substituting the values of \eqref{eq: prob-input-is-unbounded} in \eqref{eq: total-probability} completes the proof of Theorem \ref{Proposition-2: unbounded-input-channel-snr}.

\end{document}